\documentclass[a4paper,11pt]{amsart}

\usepackage{mathptmx}
\usepackage{amsthm,latexsym,amssymb,amsmath,enumerate,xcolor,wasysym,hyperref}
\usepackage[matrix,arrow]{xy}
\usepackage{pstricks,graphicx}

\renewcommand{\H}{\mathcal{H}}
\newcommand{\D}{\mathcal{D}}
\newcommand{\LL}{\mathcal{L}}
\newcommand{\R}{\mathbb{R}}
\newcommand{\N}{\mathbb{N}}
\newcommand{\Z}{\mathbb{Z}}
\newcommand{\K}{\mathbb{K}}
\newcommand{\<}{\langle}
\renewcommand{\>}{\rangle}

\renewcommand{\Re}{\mathfrak{Re}}
\renewcommand{\phi}{\varphi}
\newcommand{\na}{\nabla}
\newcommand{\dV}{\,\mathrm{dV}}
\newcommand{\ds}{\,\mathrm{dA}}
\newcommand{\grad}{\mathrm{grad}}
\renewcommand{\div}{\mathrm{div}}
\newcommand{\id}{\mathrm{id}}
\newcommand{\tr}{\mathrm{tr}}
\newcommand{\Cinf}{C^\infty}
\newcommand{\DD}{C^\infty_\mathrm{c}}
\newcommand{\Cinfsc}{C^\infty_\mathrm{sc}}
\newcommand{\supp}{\operatorname{supp}}

\renewcommand{\epsilon}{\varepsilon}
\newcommand{\ext}{\mathrm{ext}}
\newcommand{\res}{\mathrm{res}}
\newcommand{\so}[2]{\<#1,#2\>_\tau}
\newcommand{\RS}{\mathcal{Q}}
\newcommand{\ddt}{\left.\frac{d}{dt}\right|_{t=0}}
\newcommand{\ddT}[1]{\left.\frac{d}{dt_{#1}}\right|_{t_{#1}=0}}
\newcommand{\DDT}{\left.\frac{\partial^n}{\partial t_1 \cdots \partial t_n}\right|_{t_1=\cdots =t_n=0}}
\newcommand{\dtt}{\left.\frac{\partial^2}{\partial t \partial s}\right|_{t=s=0}}
\newcommand{\Siggi}{\Sigma^{3/2}M}
\renewcommand{\a}{\mathbf{a}}
\renewcommand{\b}{\mathbf{b}}
\newcommand{\Cstar}{\mbox{C$^*$}}

\newcommand{\forget}{\mathrm{REAL}}
\newcommand{\SYMPL}{\mathrm{SYMPL}}
\newcommand{\CAR}{\mathrm{CAR}}
\newcommand{\CCR}{\mathrm{CCR}}
\newcommand{\SOL}{\mathrm{SOL}}

\newcommand{\hilb}{\mathsf{HILB}}
\newcommand{\rhilb}{\mathsf{HILB}_\R}

\newcommand{\Abos}{\mathfrak{A}_\mathrm{bos}}
\newcommand{\Aferm}{\mathfrak{A}_\mathrm{ferm}}
\newcommand{\Afermsd}{\mathfrak{A}_\mathrm{ferm}^\mathrm{sd}}
\newcommand{\Sympl}{\mathsf{Sympl}}
\newcommand{\CAlg}{\mathsf{C^*Alg}}
\newcommand{\GlobHypDef}{\mathsf{GlobHypDef}}
\newcommand{\GlobHypGreen}{\mathsf{GlobHypGreen}}
\newcommand{\GHSD}{\mathsf{GlobHypSkewDef}}

\newtheorem{thm}{Theorem}[section]
\newtheorem{cor}[thm]{Corollary}
\newtheorem{prop}[thm]{Proposition}
\newtheorem{lemma}[thm]{Lemma}

\theoremstyle{definition}
\newtheorem{ex}[thm]{Example}
\newtheorem{rem}[thm]{Remark}
\newtheorem{definition}[thm]{Definition}

\newcommand{\Cl}{\operatorname{Cl}}

\def \be{\begin{eqnarray*}}
\def \ee{\end{eqnarray*}}
\def \beit{\begin{itemize}}
\def \eeit{\end{itemize}}
\def \bui#1#2{\mathrel{\mathop{\kern 0pt#1}\limits^{#2}}}
\def \buil#1#2{\mathrel{\mathop{\kern 0pt#1}\limits_{#2}}}
\def \C{\mathbb{C}}
\def \wih{\widehat}
\def \wit{\widetilde}

\title{Classical and Quantum Fields on Lorentzian Manifolds}
\author{Christian B\"ar}
\author{Nicolas Ginoux}
\address{
Universit{\"a}t Potsdam\\
Institut f\"ur Mathematik\\
Am Neuen Palais 10\\
Haus 8\\
14469 Potsdam\\
Germany
}
\address{
Fakult\"at f\"ur Mathematik\\
Universit\"at Regensburg\\
93040 Regensburg\\
Germany
}
 
\email{baer@math.uni-potsdam.de} 
\email{nicolas.ginoux@mathematik.uni-regensburg.de}

\date{\today}



\begin{document}

\begin{abstract}
We construct bosonic and fermionic locally covariant quantum field theories on curved backgrounds for large classes of fields.
We investigate the quantum field and $n$-point functions induced by suitable states.
\end{abstract}

\maketitle

\section{Introduction}
Classical fields on spacetime are mathematically modeled by sections of a vector bundle over a Lorentzian manifold.
The field equations are usually partial differential equations.
We introduce a class of differential operators,  called Green-hyperbolic operators, which have good analytical solubility properties.
This class includes wave operators as well as Dirac type operators.

In order to quantize such a classical field theory on a curved background, we need local algebras of observables.
They come in two flavors, bosonic algebras encoding the canonical commutation relations and fermionic algebras encoding the canonical anti-commutation relations.
We show how such algebras can be associated to manifolds equipped with suitable Green-hyperbolic operators.
We prove that we obtain locally covariant quantum field theories in the sense of \cite{BFV}.
There is a large literature where such constructions are carried out for particular examples of fields, see e.g.\ \cite{DHP,Dimock80,Dimock82,Furlani,Kay78,S}.
In all these papers the well-posedness of the Cauchy problem plays an important role.
We avoid using the Cauchy problem altogether and only make use of Green's operators.
In this respect, our approach is similar to the one in \cite{St}.
This allows us to deal with larger classes of fields, see Section~\ref{ssec:sum}, and to treat them systematically.
Much of the earlier work on constructing observable algebras for particular examples can be subsumed under this general approach.

It turns out that bosonic algebras can be obtained in much more general situations than fermionic algebras.
For instance, for the classical Dirac field both constructions are possible.
Hence, on the level of observable algebras, there is no spin-statistics theorem.
In order to obtain results like Theorem~5.1 in \cite{V} one needs more structure, namely representations of the observable algebras with good properties.

In order to produce numbers out of our quantum field theory that can be compared to experiments, we need states, in addition to observables.
We show how states with suitable regularity properties give rise to quantum fields and $n$-point functions.
We check that they have the properties expected from traditional quantum field theories on a Minkowski background.

{\em Acknowledgments.}
It is a pleasure to thank Alexander Strohmaier and Rainer Verch for very valuable discussion.
The authors would also like to thank SPP~1154 ``Globale Differentialgeometrie'' and SFB~647 ``Raum-Zeit-Materie'', both funded by Deutsche Forschungsgemeinschaft, for financial support.

\section{Field equations on Lorentzian manifolds}

\subsection{Globally hyperbolic manifolds}

We begin by fixing notation and recalling general facts about Lorentzian manifolds, see e.g.\ \cite{ONeill} or \cite{BGP} for more details.
Unless mentioned otherwise, the pair $(M,g)$ will stand for a smooth $m$-dimensional manifold $M$ equipped with a smooth Lorentzian metric $g$, where our convention for Lorentzian signature is $(-+\cdots+)$.
The associated volume element will be denoted by $\dV$.
We shall also assume our Lorentzian manifold $(M,g)$ to be time-orientable, i.e., that there exists a smooth timelike vector field on $M$.
Time-oriented Lorentzian manifolds will be also referred to as \emph{spacetimes}.
Note that in contrast to conventions found elsewhere, we do not assume that a spacetime is connected nor do we assume that its dimension be $m=4$.

For every subset $A$ of a spacetime $M$ we denote the causal future and past of $A$ in $M$ by $J_+(A)$ and $J_-(A)$, respectively.
If we want to emphasize the ambient space $M$ in which the causal future or past of $A$ is considered, we write $J_\pm^M(A)$ instead of $J_\pm(A)$.
Causal curves will always be implicitly assumed (future or past) oriented.

\begin{definition}\label{def-Cauchyhyp}
A \emph{Cauchy hypersurface} in a spacetime $(M,g)$ is a subset of $M$ which is met exactly once by every inextensible timelike curve.
\end{definition}
 
Cauchy hypersurfaces are always topological hypersurfaces but need not be smooth.
All Cauchy hypersurfaces of a spacetime are homeomorphic.

\begin{definition}\label{def-globhyp}
A spacetime $(M,g)$ is called \emph{globally hyperbolic} if and only if it contains a Cauchy hypersurface.
\end{definition}

A classical result of R.~Geroch \cite{Geroch70} says that a globally hyperbolic spacetime can be foliated by Cauchy hypersurfaces.
It is a rather recent and very important result that this also holds in the smooth category:

\begin{thm}[{\rm A. Bernal and M. S\'anchez \cite[Thm.~1.1]{BS}}]\label{tBernalSanchez}
Let $(M,g)$ be a globally hyperbolic spacetime.

Then there exists a smooth manifold $\Sigma$, a smooth one-parameter-family of Riemannian metrics $(g_t)_t$ on $\Sigma$ and a smooth positive function $\beta$ on $\R\times\Sigma$ such that $(M,g)$ is isometric to $(\R\times\Sigma,-\beta dt^2\oplus g_t)$.
Each $\{t\}\times\Sigma$ corresponds to a smooth spacelike Cauchy hypersurface in $(M,g)$.
\end{thm}

For our purposes, we shall need a slightly stronger version of Theorem~\ref{tBernalSanchez} where one of the Cauchy hypersurfaces $\{t\}\times\Sigma$ can be prescribed:

\begin{thm}[{\rm A. Bernal and M. S\'anchez \cite[Thm.~1.2]{BS2}}]\label{tBernalSanchez2}
Let $(M,g)$ be a globally hyperbolic spacetime and $\tilde{\Sigma}$ a smooth spacelike Cauchy hypersurface in $(M,g)$.

Then there exists a smooth splitting $(M,g)\cong(\R\times\Sigma,-\beta dt^2\oplus g_t)$ as in Theorem~\ref{tBernalSanchez} such that $\tilde{\Sigma}$ corresponds to $\{0\}\times\Sigma$.
\end{thm}

We shall also need the following result which tells us that one can extend any compact acausal spacelike submanifold to a smooth spacelike Cauchy hypersurface.
Here a subset of a spacetime is called \emph{acausal} if no causal curve meets it more than once.

\begin{thm}[{\rm A. Bernal and M. S\'anchez \cite[Thm.~1.1]{BS2}}]\label{tBernalSanchez3}
Let $(M,g)$ be a globally hyperbolic spacetime and let $K\subset M$ be a compact acausal smooth spacelike submanifold with boundary.

Then there exists a smooth spacelike Cauchy hypersurface $\Sigma$ in $(M,g)$ with $K\subset\Sigma$.
\end{thm}

\begin{definition}\label{def-spacelikecompact}
A closed subset $A\subset M$ is called {\em spacelike compact} if there exists a compact subset $K\subset M$ such that $A \subset J^M(K) := J_-^M(K) \cup J_+^M(K)$.
\end{definition}

Note that a spacelike compact subset is in general not compact, but its intersection with any Cauchy hypersurface is compact, see e.g. \cite[Cor. A.5.4]{BGP}.

\begin{definition}\label{def-causcompat}
A subset $\Omega$ of a spacetime $M$ is called \emph{causally compatible} if and only if $J_\pm^\Omega(x)=J_\pm^M(x)\cap\Omega$ for every $x\in\Omega$.
\end{definition}

This means that every causal curve joining two points in $\Omega$ must be contained entirely in $\Omega$.

\subsection{Differential operators and Green's functions}

A {\em differential operator} of order (at most) $k$ on a vector bundle $S\rightarrow M$ over $\K=\R$ or $\K=\C$ is a linear map $P:\Cinf(M,S) \to \Cinf(M,S)$ which in local coordinates $x=(x^1,\ldots,x^m)$ of $M$ and with respect to a local trivialization looks like
$$
P = \sum_{|\alpha|\leq k}A_\alpha (x) \frac{\partial^\alpha}{\partial x^\alpha}.
$$
Here $\Cinf(M,S)$ denotes the space of smooth sections of $S\rightarrow M$, $\alpha=(\alpha_1,\ldots,\alpha_m)\in \N_0\times\cdots\times\N_0$ runs over multi-indices, $|\alpha|= \alpha_1 + \ldots + \alpha_m$ and $\frac{\partial^\alpha}{\partial x^\alpha} = \frac{\partial^{|\alpha|}}{\partial (x^1)^{\alpha_1} \cdots \partial (x^m)^{\alpha_m}}$.
The {\em principal symbol} $\sigma_P$ of $P$ associates to each covector $\xi\in T^*_xM$ a linear map $\sigma_P(\xi):S_x \to S_x$.
Locally, it is given by
$$
\sigma_P(\xi) = \sum_{|\alpha|= k}A_\alpha (x) \xi^\alpha
$$
where $\xi^\alpha = \xi_1^{\alpha_1}\cdots \xi_m^{\alpha_m}$ and $\xi = \sum_j \xi_j dx^j$.
If $P$ and $Q$ are two differential operators of order $k$ and $\ell$ respectively, then $Q\circ P$ is a differential operator of order $k+\ell$ and 
$$
\sigma_{Q\circ P}(\xi) = \sigma_{Q}(\xi) \circ \sigma_{P}(\xi) .
$$
For any linear differential operator $P:\Cinf(M,S) \to \Cinf(M,S)$ there is a unique formally dual operator $P^*:\Cinf(M,S^*) \to \Cinf(M,S^*)$ of the same order characterized by 
$$
\int_M \<\phi,P\psi\> \dV = \int_M \<P^*\phi,\psi\> \dV 
$$
for all $\psi\in\Cinf(M,S)$ and $\phi\in\Cinf(M,S^*)$ with $\supp(\phi) \cap \supp(\psi)$ compact.
Here $\<\cdot,\cdot\>:S^*\otimes S \to \K$ denotes the canonical pairing, i.e., the evaluation of a linear form in $S_x^*$ on an element of $S_x$, where $x\in M$.
We have $\sigma_{P^*}(\xi) = (-1)^k\sigma_P(\xi)^*$ where $k$ is the order of $P$.

\begin{definition}\label{def-selfadjoint}
Let a vector bundle $S\to M$ be endowed with a non-degenerate inner product $\<\cdot\,,\cdot\>$.
A linear differential operator $P$ on $S$ is called \emph{formally self-adjoint} if and only if
\[ \int_M\< P\phi,\psi\> \dV =\int_M\<\phi,P\psi\> \dV\]
holds for all $\phi,\psi\in \Cinf(M,S)$ with $\supp(\phi) \cap \supp(\psi)$ compact.

Similarly, we call $P$ \emph{formally skew-adjoint} if instead
\[ \int_M\< P\phi,\psi\> \dV =-\int_M\<\phi,P\psi\> \dV\, .\]
\end{definition}

We recall the definition of advanced and retarded Green's operators for a linear differential operator.

\begin{definition}\label{def-Green}
Let $P$ be a linear differential operator acting on the sections of a vector bundle $S$ over a Lorentzian manifold $M$.
An \emph{advanced Green's operator} for $P$ on $M$ is a linear map 
\[G_+:\DD(M,S)\to \Cinf(M,S)\]
satisfying:
\begin{itemize}
\item[(G$_1$)] $P\circ G_+=\id_{_{\DD(M,S)}}$;
\item[(G$_2$)] $G_+\circ P_{|_{\DD(M,S)}}=\id_{_{\DD(M,S)}}$;
\item[(G$_3^+$)] $\supp(G_+\phi)\subset J_+^M(\supp(\phi))$ for any $\phi\in\DD(M,S)$.
\end{itemize}
A \emph{retarded Green's operator} for $P$ on $M$ is a linear map $G_-:\DD(M,S)\to \Cinf(M,S)$ satisfying (G$_1$), (G$_2$), and 
\begin{itemize}
\item[(G$_3^-$)]
$\supp(G_-\phi)\subset J_-^M(\supp(\phi))\textrm{ for any }\phi\in\DD(M,S).$
\end{itemize}
\end{definition}

Here we denote by $\DD(M,S)$ the space of compactly supported smooth sections of $S$.

\begin{definition}\label{def-Greenhyperbolic}
Let $P:\Cinf(M,S) \to \Cinf(M,S)$ be a linear differential operator.
We call $P$ \emph{Green-hyperbolic} if the restriction of $P$ to any globally hyperbolic subregion of $M$ has advanced and retarded Green's operators.\end{definition}

\begin{rem}\label{rem:GreenUnique1}
If the Green's operators of the restriction of $P$ to a globally hyperbolic subregion exist, then they are necessarily unique, see Remark~\ref{rem:GreenUnique}.
\end{rem}

\subsection{Wave operators}
The most prominent class of Green-hyperbolic operators are wave operators, sometimes also called normally hyperbolic operators.

\begin{definition}
A linear differential operator of second order $P:\Cinf(M,S) \to \Cinf(M,S)$ is called a \emph{wave operator} if its principal symbol is given by the Lorentzian metric, i.e., for all $\xi\in T^*M$ we have
$$
\sigma_P(\xi) = -\<\xi,\xi\>\cdot\id .
$$
\end{definition}

In other words, if we choose local coordinates $x^1,\ldots,x^m$ on $M$ and a
local trivialization of $S$, then 
$$
P = -\sum_{i,j=1}^m g^{ij}(x)\frac{\partial^2}{\partial x^i \partial x^j}
+ \sum_{j=1}^m A_j(x) \frac{\partial}{\partial x^j} + B(x)
$$
where $A_j$ and $B$ are matrix-valued coefficients depending smoothly on
$x$ and $(g^{ij})$ is the inverse matrix of  $(g_{ij})$ with
$g_{ij}=\<\frac{\partial}{\partial x^i},\frac{\partial}{\partial x^j}\>$.
If $P$ is a wave operator, then so is its dual operator $P^*$.
In \cite[Cor.~3.4.3]{BGP} it has been shown that wave operators are Green-hyperbolic.

\begin{ex}[d'Alembert operator]\label{ex:dAlembert}
Let $S$ be the trivial line bundle so that sections of $S$ are just functions.
The d'Alembert operator $P=\Box=-\div\circ\grad$ is a formally self-adjoint wave operator, see e.g.\ \cite[p.~26]{BGP}.
\end{ex}

\begin{ex}[connection-d'Alembert operator]\label{ex:boxnabla}
More generally, let $S$ be a vector bundle and let $\nabla$ be a connection on $S$.
This connection and the Levi-Civita connection on $T^*M$ induce a connection on $T^*M\otimes S$, again denoted $\nabla$.
We define the connection-d'Alembert operator $\Box^\nabla$ to be the composition of the following three maps
$$
\Cinf(M,S) \xrightarrow{\nabla}
\Cinf(M,T^*M\otimes S) \xrightarrow{\nabla}
\Cinf(M,T^*M\otimes T^*M\otimes S)
\xrightarrow{-\tr\otimes\id_S} 
\Cinf(M,S)
$$
where $\tr:T^*M\otimes T^*M \to \R$ denotes the metric trace,
$\tr(\xi\otimes\eta)=\<\xi,\eta\>$. 
We compute the principal symbol,
$$
\sigma_{\Box^\nabla}(\xi)\phi
=
-(\tr\otimes\id_S)\circ\sigma_{\nabla}(\xi)\circ\sigma_{\nabla}(\xi)(\phi)
=
-(\tr\otimes\id_S)(\xi\otimes\xi\otimes\phi)
=
-\<\xi,\xi\>\,\phi.
$$
Hence $\Box^\nabla$ is a wave operator.
\end{ex}

\begin{ex}[Hodge-d'Alembert operator]\label{ex:kforms}
Let $S=\Lambda^kT^*M$ be the bundle of $k$-forms.
Exterior differentiation $d:\Cinf(M,\Lambda^kT^*M) \to
\Cinf(M,\Lambda^{k+1}T^*M)$ increases the degree by one while
the codifferential $\delta=d^*:\Cinf(M,\Lambda^{k}T^*M) \to
\Cinf(M,\Lambda^{k-1}T^*M)$ decreases the degree by one.
While $d$ is independent of the metric, the codifferential $\delta$ does
depend on the Lorentzian metric.
The operator $P=-d\delta - \delta d$ is a formally self-adjoint wave operator.
\end{ex}

\subsection{The Proca equation}

The Proca operator is an example of a Green-hyperbolic operator of second order which is not a wave operator.
First we need the following observation:

\begin{lemma}
Let $M$ be globally hyperbolic, let $S\to M$ be a vector bundle and let $P$ and $Q$ be differential operators acting on sections of $S$.
Suppose $P$ has advanced and retarded Green's operators $G_+$ and $G_-$.

If $Q$ commutes with $P$, then it also commutes with $G_+$ and with $G_-$.
\end{lemma}

\begin{proof}
Assume $[P,Q]=0$.
We consider
$$
\tilde G_\pm := G_\pm + [G_\pm,Q] :\,\, \DD(M,s) \to \Cinfsc(M,S).
$$
We compute on $\DD(M,S)$:
$$
\tilde G_\pm P 
= G_\pm P + G_\pm QP -QG_\pm P 
= \id + G_\pm PQ - Q
= \id + Q - Q 
= \id 
$$
and similarly $P\tilde G_\pm = \id$.
Hence $\tilde G_\pm$ are also advanced and retarded Green's operators, respectively.
By Remark~\ref{rem:GreenUnique1}, Green's operators are unique, hence $\tilde G_\pm=G_\pm$ and therefore $[G_\pm,Q]=0$.
\end{proof}

\begin{ex}[Proca operator]
The discussion of this example follows \cite[p.~116f]{St}, see also \cite{Furlani} where is the discussion is based on the Cauchy problem.
The Proca equation describes massive vector bosons.
We take $S=T^*M$ and let $m_0>0$.
The Proca equation is
\begin{equation}
P\phi := \delta d \phi + m_0^2 \phi =0
\label{eq:Proca}
\end{equation}
where $\phi\in\Cinf(M,S)$.
Applying $\delta$ to \eqref{eq:Proca} we obtain, using $\delta^2=0$ and $m_0\neq 0$,
\begin{equation}
\delta\phi=0
\label{eq:Proca2}
\end{equation}
and hence
\begin{equation}
(d\delta+\delta d)\phi + m_0^2\phi = 0.
\label{eq:Proca3}
\end{equation}
Conversely, \eqref{eq:Proca2} and \eqref{eq:Proca3} clearly imply \eqref{eq:Proca}.

Since $\tilde P := d\delta+\delta d + m_0^2$ is minus a wave operator, it has Green's operators $\tilde G_\pm$.
We define
$$
G_\pm :\DD(M,S) \to \Cinfsc(M,S), \quad 
G_\pm 
:= (m_0^{-2}d\delta + \id)\circ \tilde G_\pm 
= \tilde G_\pm \circ (m_0^{-2}d\delta + \id) \, .
$$
The last equality holds because $d$ and $\delta$ commute with $\tilde P$.
For $\phi\in\DD(M,S)$ we compute
$$
G_\pm P\phi 
=
\tilde G_\pm (m_0^{-2}d\delta + \id)(\delta d + m_0^2)\phi=
\tilde G_\pm \tilde P\phi = \phi
$$
and similarly $PG_\pm\phi=\phi$.
Since the differential operator $m_0^{-2}d\delta + \id$ does not increase supports, the third axiom in the definition of advanced and retarded Green's operators holds as well.

This shows that $G_+$ and $G_-$ are advanced and retarded Green's operators for $P$, respectively.
Thus $P$ is not a wave operator but Green-hyperbolic.
\end{ex}

\subsection{Dirac type operators}\label{s:Diractype}
The most important Green-hyperbolic operators of first order are the so-called Dirac type operators.

\begin{definition}
A linear differential operator $D:\Cinf(M,S) \to \Cinf(M,S)$ of first order is called {\em of Dirac type}, if $-D^2$ is a wave operator.
\end{definition}

\begin{rem}
If $D$ is of Dirac type, then $i$ times its principal symbol satisfies the Clifford relations
$$
(i\sigma_D(\xi))^2 = -\sigma_{D^2}(\xi) = -\<\xi,\xi\>\cdot\id ,
$$
hence by polarization
$$
(i\sigma_D(\xi))(i\sigma_D(\eta)) + (i\sigma_D(\eta))(i\sigma_D(\xi)) = -2\<\xi,\eta\>\cdot\id .
$$
The bundle $S$ thus becomes a module over the bundle of Clifford algebras $\Cl(TM)$ associated with $(TM,\<\cdot\,,\cdot\>)$.
See \cite[Sec.~1.1]{Baum} or \cite[Ch.~I]{LM} for the definition and properties of the Clifford algebra $\Cl(V)$ associated with a vector space $V$ with inner product.
\end{rem}

\begin{rem}
If $D$ is of Dirac type, then so is its dual operator $D^*$.
On a globally hyperbolic region let $G_+$ be the advanced Green's operator for $D^2$ which exists since $-D^2$ is a wave operator.
Then it is not hard to check that $D\circ G_+$ is an advanced Green's operator for $D$, see e.g.\ the proof of Theorem~2.3 in \cite{DHP} or \cite[Thm.~3.2]{Mue}.
The same discussion applies to the retarded Green's operator.
Hence any Dirac type operator is Green-hyperbolic.
\end{rem}

\begin{ex}[Classical Dirac operator]\label{ex:spinors}
If the spacetime $M$ carries a spin structure, then one can define the spinor bundle $S=\Sigma M$ and the classical Dirac operator
$$
D : \Cinf(M,\Sigma M) \to \Cinf(M,\Sigma M) ,\quad D\phi:=i\sum_{j=1}^m\epsilon_j e_j\cdot\na_{e_j}\phi .
$$
Here $(e_j)_{1\leq j\leq m}$ is a local orthonormal basis of the tangent bundle, $\epsilon_j=\<e_j,e_j\>=\pm 1$ and ``$\cdot$'' denotes the Clifford multiplication, see e.g.\ \cite{Baum} or \cite[Sec.~2]{BGM}.
The principal symbol of $D$ is given by 
$$
\sigma_D(\xi)\psi = i\xi^\sharp \cdot \psi.
$$
Here $\xi^\sharp$ denotes the tangent vector dual to the 1-form $\xi$ via the Lorentzian metric, i.e., $\<\xi^\sharp,Y\> = \xi(Y)$ for all tangent vectors $Y$ over the same point of the manifold.
Hence
$$
\sigma_{D^2}(\xi)\psi = \sigma_D(\xi)\sigma_D(\xi)\psi = -\xi^\sharp \cdot
\xi^\sharp \cdot \psi = \<\xi,\xi\>\,\psi.
$$
Thus $P=-D^2$ is a wave operator.
Moreover, $D$ is formally self-adjoint, see e.g. \cite[p.~552]{BGM}.
\end{ex}

\begin{ex}[Twisted Dirac operators]\label{ex:twistedDirac}
More generally, let $E\to M$ be a complex vector bundle equipped with a non-degenerate Hermitian inner product and a metric connection $\nabla^E$ over a spin spacetime $M$.
In the notation of Example~\ref{ex:spinors}, one may define the Dirac operator of $M$ twisted with $E$ by $$
D^E:=i\sum_{j=1}^m\epsilon_j e_j\cdot\na_{e_j}^{\Sigma M\otimes E}:\Cinf(M,\Sigma M\otimes E)\to\Cinf(M,\Sigma M\otimes E),
$$
where $\na^{\Sigma M\otimes E}$ is the tensor product connection on $\Sigma M\otimes E$.
Again, $D^E$ is a formally self-adjoint Dirac type operator.
\end{ex}

\begin{ex}[Euler operator]\label{ex:Euler}
In Example~\ref{ex:kforms}, replacing $\Lambda^kT^*M$ by $S:=\Lambda T^*M\otimes\C = \oplus_{k=0}^n\Lambda^kT^*M\otimes\C$, the Euler operator $D=i(d-\delta)$ defines a formally self-adjoint Dirac type operator.
In case $M$ is spin, the Euler operator coincides with the Dirac operator of $M$ twisted with $\Sigma M$.
\end{ex}

\begin{ex}[Buchdahl operators]\label{ex:Buchdahl}
On a $4$-dimensional spin spacetime $M$, consider the standard orthogonal and parallel splitting $\Sigma M=\Sigma_+M\oplus\Sigma_-M$ of the complex spinor bundle of $M$ into spinors of positive and negative chirality.
The finite dimensional irreducible representations of the simply-connected Lie group $\mathrm{Spin}^0(3,1)$ are given by $\Sigma_+^{({k}/{2})} \otimes \Sigma_-^{({\ell}/{2})}$ where $k,\ell\in\N$.
Here $\Sigma_+^{({k}/{2})}=\Sigma_+^{\odot k}$ is the $k$-th symmetric tensor product of the positive half-spinor representation $\Sigma_+$ and similarly for $\Sigma_-^{({\ell}/{2})}$.
Let the associated vector bundles $\Sigma_\pm^{({k}/{2})} M$ carry the induced inner product and connection.

For $s\in\mathbb{N}$, $s\geq1$, consider the twisted Dirac operator $D^{(s)}$ acting on sections of $\Sigma M\otimes\Sigma_+^{({(s-1)}/{2})}M$.
In the induced splitting 
$$
\Sigma M\otimes\Sigma_+^{({(s-1)}/{2})}M
=
\Sigma_+M\otimes\Sigma_+^{({s-1}/{2})}M\oplus\Sigma_-M\otimes\Sigma_+^{({(s-1)}/{2})}M
$$ 
the operator $D^{(s)}$ is of the form 
$$
\left(\begin{array}{cc}0&D_-^{(s)}\\ D_+^{(s)}&0\end{array}\right)
$$
because Clifford multiplication by vectors exchanges the chiralities.
The Clebsch-Gordan formulas \cite[Prop. II.5.5]{BrtD} tell us that the representation $\Sigma_+\otimes\Sigma_+^{(\frac{s-1}{2})}$ splits as 
$$
\Sigma_+\otimes\Sigma_+^{(\frac{s-1}{2})}
=
\Sigma_+^{(\frac{s}{2})}\oplus\Sigma_+^{(\frac{s}{2}-1)}.
$$
Hence we have the corresponding parallel orthogonal projections 
$$
\pi_s :\Sigma_+M\otimes\Sigma_+^{(\frac{s-1}{2})}M \to \Sigma_+^{(\frac{s}{2})}M
\quad\mbox{ and }\quad
\pi_s' :\Sigma_+M\otimes\Sigma_+^{(\frac{s-1}{2})}M \to \Sigma_+^{(\frac{s}{2}-1)}M .
$$
On the other hand, the representation $\Sigma_-\otimes\Sigma_+^{(\frac{s-1}{2})}$ is irreducible.
Now \emph{Buchdahl operators} are the operators of the form 
$$
B_{\mu_1,\mu_2,\mu_3}^{(s)}
:=
\left(\begin{array}{cc}\mu_1\cdot\pi_s + \mu_2\cdot\pi_s'&D_-^{(s)}\\ D_+^{(s)}&\mu_3\cdot\id\end{array}\right)
$$
where $\mu_1,\mu_2,\mu_3\in\mathbb{C}$ are constants.
By definition, $B_{\mu_1,\mu_2,\mu_3}^{(s)}$ is of the form $D^{(s)}+b$, where $b$ is of order zero.
In particular, $B_{\mu_1,\mu_2,\mu_3}^{(s)}$ is a Dirac-type operator, hence it is Green-hyperbolic.

If $M$ were Riemannian, then $D^{(s)}$ would be formally self-adjoint. Hence the operator $B_{\mu_1,\mu_2,\mu_3}^{(s)}$ would be formally self-adjoint if and only if the constants $\mu_1,\mu_2,\mu_3$ are real.
In Lorentzian signature, $\Sigma_+M$ and $\Sigma_-M$ are isotropic for the natural inner product on $\Sigma M$, so that the bundles on which the Buchdahl operators act, carry no natural non-degenerate inner product.

For a definition of Buchdahl operators using indices we refer to \cite{Buchdahl2,Buchdahl3,Wue} and to \cite[Def.~8.1.4, p.~104]{MueDipl}.
\end{ex}

\subsection{The Rarita-Schwinger operator}\label{ss:RaritaSchwinger}

For the Rarita-Schwinger operator on Riemannian manifolds, we refer to \cite[Sec.~2]{Wang}, see also \cite[Sec.~2]{BransonHijaziBW}.
In this section let the spacetime $M$ be spin and consider the Clifford-multiplication $\gamma:T^*M\otimes\Sigma M\rightarrow \Sigma M$, $\theta\otimes\psi\mapsto \theta^\sharp\cdot\psi$, where $\Sigma M$ is the complex spinor bundle of $M$.
Then there is the representation theoretic splitting of $T^*M\otimes\Sigma M$ into the orthogonal and parallel sum
\[
T^*M\otimes\Sigma M=\iota(\Sigma M)\oplus \Siggi,
\]
where $\Siggi:=\mathrm{ker}(\gamma)$ and $\iota(\psi):=-\frac{1}{m}\sum_{j=1}^me_j^*\otimes e_j\cdot\psi$.
Here again $(e_j)_{1\leq j\leq m}$ is a local orthonormal basis of the tangent bundle.
Let $\mathcal{D}$ be the twisted Dirac operator on $T^*M\otimes\Sigma M$, that is, $\mathcal{D}:=i\cdot(\id\otimes\gamma)\circ\nabla$,
where $\nabla$ denotes the induced covariant derivative on $T^*M\otimes\Sigma M$.

\begin{definition}
The \emph{Rarita-Schwinger operator} on the spin spacetime $M$ is defined by $\RS:=(\id-\iota\circ\gamma)\circ\mathcal{D}:\Cinf(M,\Siggi)\rightarrow\Cinf(M,\Siggi)$.
\end{definition}

By definition, the Rarita-Schwinger operator is pointwise obtained as the orthogonal projection onto $\Siggi$ of the twisted Dirac operator $\mathcal{D}$ restricted to a section of $\Siggi$. 
Using the above formula for $\iota$, the Rarita-Schwinger operator can be written down explicitly:
\[
\RS\psi
=
i\cdot\sum_{\beta=1}^me_\beta^*\otimes \sum_{\alpha=1}^m\epsilon_\alpha(e_\alpha\cdot\nabla_{e_\alpha}\phi_\beta-\frac{2}{m}e_\beta\cdot\nabla_{e_\alpha}\phi_\alpha)
\]
for all $\psi=\sum_{\beta=1}^me_\beta^*\otimes\psi_\beta\in\Cinf(M,\Siggi)$, where here $\nabla$ is the standard connection on $\Sigma M$.
It can be checked that $\RS$ is a formally self-adjoint linear differential operator of first order, with principal symbol
\[
\sigma_\RS(\xi):\psi\mapsto i\Big\{(\id\otimes\xi^\sharp\cdot)\psi-\frac{2}{m}\sum_{\beta=1}^me_\beta^*\otimes e_\beta\cdot(\xi^\sharp\lrcorner\psi)\Big\},
\]
for all $\psi=\sum_{\beta=1}^me_\beta^*\otimes\psi_\beta\in \Siggi$.
Here $X\lrcorner\psi$ denotes the insertion of the tangent vector $X$ in the first factor, that is, $X\lrcorner\psi:=\sum_{\beta=1}^me_\beta^*(X)\psi_\beta$.\\

\begin{lemma}\label{l:RSnotdef}
Let $M$ be a spin spacetime of dimension $m\ge 3$.
Then the characteristic variety of the Rarita-Schwinger operator of $M$ coincides with the set of lightlike covectors.
\end{lemma}

\begin{proof}
By definition, the characteristic variety of $\RS$ is the set of nonzero co\-vectors $\xi$ for which $\sigma_\RS(\xi)$ is not invertible.
Fix an arbitrary point $x\in M$.
Let $\xi\in T_x^*M\setminus\{0\}$ be non-lightlike.
Without loss of generality we may assume that $\xi$ is normalized and that the Lorentz orthonormal basis is chosen so that $\xi^\sharp=e_1$.
Hence $\epsilon_1=1$ if $\xi$ is spacelike and $\epsilon_1=-1$ if $\xi$ is timelike.
Take $\psi=\sum_{\beta=1}^me_\beta^*\otimes\psi_\beta\in\mathrm{ker}(\sigma_\RS(\xi))$.
Then
\be 
0&=&\sum_{\beta=1}^me_\beta^*\otimes e_1\cdot\psi_\beta-\frac{2}{m}\sum_{\beta=1}^me_\beta^*\otimes e_\beta\cdot\psi_1\\
&=&\sum_{\beta=1}^me_\beta^*\otimes (e_1\cdot\psi_\beta-\frac{2}{m}e_\beta\cdot\psi_1),
\ee
which implies $e_1\cdot\psi_\beta=\frac{2}{m}e_\beta\cdot\psi_1$ for all $\beta\in\{1,\ldots,m\}$.
Choosing $\beta=1$, we obtain $e_1\cdot\psi_1=0$ because $m\geq 3$.
Hence $\psi_1=0$, from which $\psi_\beta=0$ follows for all $\beta\in\{1,\ldots,m\}$.
Hence $\psi=0$ and $\sigma_\RS(\xi)$ is invertible.

If $\xi\in T_x^*M\setminus\{0\}$ is lightlike, then we may assume that $\xi^\sharp=e_1+e_2$, where $\epsilon_1=-1$ and $\epsilon_2=1$.
Choose $\psi_1\in\Sigma_xM\setminus\{0\}$ with $(e_1+e_2)\cdot\psi_1=0$.
Such a $\psi_1$ exists because Clifford multiplication by a lightlike vector is nilpotent.
Set $\psi_2:=-\psi_1$ and $\psi:=e_1^*\otimes\psi_1+e_2^*\otimes\psi_2$.
Then $\psi\in \Sigma^{3/2}_xM\setminus\{0\}$ and
\[
-i\sigma_\RS(\xi)(\psi)=\sum_{j=1}^2e_j^*\otimes \underbrace{(e_1+e_2)\cdot\psi_j}_{=0}-\frac{2}{m}e_j^*\otimes e_j\cdot(\underbrace{\psi_1+\psi_2}_{=0})
=0.
\]
This shows $\psi\in\mathrm{ker}(\sigma_\RS(\xi))$ and hence $\sigma_\RS(\xi)$ is not invertible.
\end{proof}

The same proof shows that in the Riemannian case the Rarita-Schwinger operator is elliptic.

\begin{rem}
Since the characteristic variety of the Rarita-Schwinger operator is exactly that of the Dirac operator, Lemma~\ref{l:RSnotdef} together with \cite[Thms. 23.2.4 \& 23.2.7]{Hoerm3} imply that the Cauchy problem for $\RS$ is well-posed in case $M$ is globally hyperbolic.
This implies they $\RS$ has advanced and retarded Green's operators.
Hence $\RS$ is not of Dirac type but it is Green-hyperbolic.
\end{rem}

\begin{rem}
The equations originally considered by Rarita and Schwinger in \cite{RS} correspond to the twisted Dirac operator $\mathcal{D}$ restricted to $\Siggi$ but not projected back to $\Siggi$.
In other words, they considered the operator
$$
\mathcal{D}|_{C^\infty(M,\Siggi)} : C^\infty(M,\Siggi) \to C^\infty(M,T^*M\otimes \Sigma M) .
$$
These equations are over-determined.
Therefore it is not a surprise that non-trivial solutions restrict the geometry of the underlying manifold as observed by Gibbons \cite{Gib} and that this operator has no Green's operators.
\end{rem}

\subsection{Combining given operators into a new one}\label{ssec:sum}
Given two Green-hyperbolic operators we can form the direct sum and obtain a new operator in a trivial fashion.
It turns out that this operator is again Green-hyperbolic.
Note that the two operators need not have the same order.

\begin{lemma}
Let $S_1, S_2 \to M$ be two vector bundles over the globally hyperbolic manifold $M$.
Let $P_1$ and $P_2$ be two Green-hyperbolic operators acting on sections of $S_1$ and $S_2$ respectively.
Then 
$$
P_1 \oplus P_2 := \begin{pmatrix}
                         P_1 & 0 \cr
                         0   & P_2
                        \end{pmatrix}
: \Cinf(M,S_1\oplus S_2) \to \Cinf(M,S_1\oplus S_2)
$$
is Green-hyperbolic.
\end{lemma}

\begin{proof}
If $G_1$ and $G_2$ are advanced Green's operators for $P_1$ and $P_2$ respectively, then clearly $\begin{pmatrix}G_1 & 0 \cr 0   & G_2\end{pmatrix}$ is an advanced Green's operator for $P_1 \oplus P_2$.
The retarded case is analogous.
\end{proof}

It is interesting to note that $P_1$ and $P_2$ need not have the same order.
Hence Green-hyperbolic operators need not be hyperbolic in the usual sense.
Moreover, it is not obvious that Green-hyperbolic operators have a well-posed Cauchy problem.
For instance, if $P_1$ is a wave operator and $P_2$ a Dirac-type operator, then along a Cauchy hypersurface one would have to prescribe the normal derivative for the $S_1$-component but not for the $S_2$-component.

\section{Algebras of observables}

Our next aim is to quantize the classical fields governed by Green-hyperbolic differential operators.
We construct local algebras of observables and we prove that we obtain locally covariant quantum field theories in the sense of \cite{BFV}.

\subsection{Bosonic quantization}\label{ss:CCRquant}

In this section we show how a quantization process based on canonical commutation relations (CCR) can be carried out for formally self-adjoint Green-hyperbolic operators.
This is a functorial procedure.
We define the first category involved in the quantization process.

\begin{definition}\label{def:GHG}
The category $\GlobHypGreen$ consists of the following objects and morphisms:
\beit\item An object in $\GlobHypGreen$ is a triple $(M,S,P)$, where 
\beit\item[$\RHD$]
$M$ is a globally hyperbolic spacetime, 
\item[$\RHD$]
$S$ is a real vector bundle over $M$ endowed with a non-degenerate inner product $\<\cdot\,,\cdot\>$ and 
\item[$\RHD$]
$P$ is a formally self-adjoint Green-hyperbolic operator acting on sections of $S$.
\eeit
\item A morphism between two objects $(M_1,S_1,P_1)$ and $(M_2,S_2,P_2)$ of $\GlobHypGreen$ is a pair $(f,F)$, where 
\beit\item[$\RHD$] 
$f$ is a time-orientation preserving isometric embedding $M_1\rightarrow M_2$ with $f(M_1)$ causally compatible and open in $M_2$, 
\item[$\RHD$] 
$F$ is a fiberwise isometric vector bundle isomorphism over $f$ such that the following diagram commutes:
\begin{equation}\label{diag:restrictD}
\xymatrix{\Cinf(M_2,S_2)\ar[r]^{P_2}\ar[dd]_{\res}&\Cinf(M_2,S_2)\ar[dd]_{\res}\\ &\\ \Cinf(M_1,S_1)\ar[r]^{P_1} &\Cinf(M_1,S_1),}
\end{equation}
where $\res(\phi):=F^{-1}\circ\phi\circ f$ for every $\phi\in \Cinf(M_2,S_2)$. 
\eeit
\eeit
\end{definition}

Note that morphisms exist only if the manifolds have equal dimension and the vector bundles have the same rank. 
Note furthermore, that the inner product $\<\cdot\,,\cdot\>$ on $S$ is not required to be positive or negative definite.

The causal compatibility condition, which is not automatically satisfied (see e.g. \cite[Fig. 33]{BGP}), ensures the commutation of the extension and restriction maps with the Green's operators:

\begin{lemma}\label{l:extcommutesGreen}
Let $(f,F)$ be a morphism between two objects $(M_1,S_1,P_1)$ and $(M_2,S_2,P_2)$ in the category $\GlobHypGreen$ and let $(G_1)_\pm$ and $(G_2)_\pm$ be the respective Green's operators for $P_1$ and $P_2$. 
Denote by $\ext(\phi)\in\DD(M_2,S_2)$ the extension by $0$ of $F\circ\phi\circ f^{-1}:f(M_1)\rightarrow S_2$ to $M_2$, for every $\phi\in\DD(M_1,S_1)$. 
Then 
\[\res\circ (G_2)_\pm\circ\ext=(G_1)_\pm.\]
\end{lemma}

\begin{proof}
Set $(\wit{G_1})_\pm:=\res\circ (G_2)_\pm\circ\ext$ and fix $\phi\in\DD(M_1,S_1)$.
First observe that the causal compatibility condition on $f$ implies that 
\be 
\supp((\wit{G_1})_\pm(\phi))&=&f^{-1}(\supp((G_2)_\pm\circ\ext(\phi)))\\
&\subset&f^{-1}(J_\pm^{M_2}(\supp(\ext(\phi))))\\
&=&f^{-1}(J_\pm^{M_2}(f(\supp(\phi))))\\
&=&J_\pm^{M_1}(\supp(\phi)).
\ee
In particular, $(\wit{G_1})_\pm(\phi)$ has spacelike compact support in $M_1$ and $(\wit{G_1})_\pm$ satisfies Axiom~${\rm (G_3)}$. 
Moreover, it follows from (\ref{diag:restrictD}) that $P_2\circ\ext=\ext\circ P_1$ on $\DD(M_1,S_1)$, which directly implies that $(\wit{G_1})_\pm$ satisfies Axioms ${\rm (G_1)}$ and ${\rm (G_2)}$ as well. 
The uniqueness of the advanced and retarded Green's operators on $M_1$ yields $(\wit{G_1})_\pm=(G_1)_\pm$.
\end{proof}

Next we show how the Green's operators for a formally self-adjoint Green-hyperbolic operator provide a symplectic vector space in a canonical way.
First we see how the Green's operators of an operator and of its formally dual operator are related.

\begin{lemma}\label{l:GpmDsa}
Let $M$ be a globally hyperbolic spacetime and $G_+,G_-$ the advanced and retarded Green's operators for a Green-hyperbolic operator $P$ acting on sections of $S\rightarrow M$.
Then the advanced and retarded Green's operators $G_+^*$ and $G_-^*$ for $P^*$ satisfy
\[
\int_M\< G_\pm^*\phi,\psi\>\dV=\int_M\<\phi,G_\mp\psi\>\dV 
\]
for all $\phi\in \DD(M,S^*)$ and $\psi\in \DD(M,S)$. 
\end{lemma}

\begin{proof}
Axiom ${\rm (G_1)}$ for the Green's operators implies that
\be 
\int_M\< G_\pm^*\phi,\psi\>\dV&=&\int_M\<G_\pm^*\phi,P(G_\mp\psi)\>\dV\\
&=&\int_M\<P^*(G_\pm^*\phi),G_\mp \psi\>\dV\\
&=&\int_M\<\phi,G_\mp \psi\>\dV,
\ee
where the integration by parts is justified since $\supp(G_\pm^*\phi)\cap\supp(G_\mp\psi)\subset J_\pm^M(\supp(\phi))\cap J_\mp^M(\supp(\psi))$ is compact.
\end{proof}

\begin{prop}\label{p:symplstruct}
Let $(M,S,P)$ be an object in the category $\GlobHypGreen$.
Set $G:=G_+-G_-$, where $G_+, G_-$ are the advanced and retarded Green's operator for $P$, respectively.

Then the pair $(\SYMPL(M,S,P),\omega)$ is a symplectic vector space, where
\[
\SYMPL(M,S,P):=\DD(M,S)/\ker(G)
\quad\mbox{ and }\quad
\omega([\phi],[\psi]):=\int_M\< G\phi,\psi\> \dV.
\]
Here the square brackets $[\cdot]$ denote residue classes modulo $\ker(G)$.
\end{prop}

\begin{proof}
The bilinear form $(\phi,\psi)\mapsto\int_M\< G\phi,\psi\> \dV$ on $\DD(M,S)$ is skew-symmetric as a consequence of Lemma~\ref{l:GpmDsa} because $P$ is formally self-adjoint. 
Its null-space is exactly $\ker(G)$. 
Therefore the induced bilinear form $\omega$ on the quotient space $\SYMPL(M,S,P)$ is non-degenerate and hence a symplectic form.
\end{proof}

Put $\Cinfsc(M,S):=\{\phi\in \Cinf(M,S)\,|\,\mathrm{supp}(\phi)\textrm{ is spacelike compact}\}$.
The next result will in particular show that we can consider $\SYMPL(M,S,P)$ as the space of smooth solutions of the equation $P\phi=0$ which have spacelike compact support.

\begin{thm}\label{pexactseq}
Let $M$ be a Lorentzian manifold, let $S\to M$ be a vector bundle, and let $P$ be a Green-hyperbolic operator acting on sections of $S$.
Let $G_\pm$ be advanced and retarded Green's operators for $P$, respectively.
Put 
$$
G:=G_+-G_-: \DD(M,S) \to \Cinfsc(M,S).
$$
Then the following linear maps form a complex:
\begin{equation}\label{exactseq}
 \{0\}\to\DD(M,S)\xrightarrow{P}\DD(M,S)\xrightarrow{G}\Cinfsc(M,S)\xrightarrow{P}\Cinfsc(M,S).
\end{equation}
This complex is always exact at the first $\DD(M,S)$.
If $M$ is globally hyperbolic, then the complex is exact everywhere.
\end{thm}

\begin{proof}
The proof follows the lines of \cite[Thm.~3.4.7]{BGP} where the result was shown for wave operators.
First note that, by (G$_3^\pm$) in the definition of Green's operators, we have that $G_\pm : \DD(M,S) \to \Cinfsc(M,S)$.
It is clear from (G$_1$) and (G$_2$) that $PG=GP=0$ on $\DD(M,S)$, hence \eqref{exactseq} is a complex.

If $\phi\in\DD(M,S)$ satisfies $P\phi=0$, then by (G$_2$) we have $\phi=G_+P\phi=0$ which shows that $P_{|_{\DD(M,S)}}$ is injective.
Thus the complex is exact at the first $\DD(M,S)$.

From now on let $M$ be globally hyperbolic.
Let $\phi\in\DD(M,S)$ with $G\phi=0$, i.e., $G_+\phi=G_-\phi$. 
We put $\psi:=G_+\phi=G_-\phi \in \Cinf(M,S)$ and we see that
$\supp(\psi)=\supp(G_+\phi)\cap\supp(G_-\phi)\subset J_+(\supp(\phi))\cap 
J_-(\supp(\phi))$.
Since $(M,g)$ is globally hyperbolic $J_+(\supp(\phi))\cap
J_-(\supp(\phi))$ is compact, hence $\psi\in \DD(M,S)$.
From $P\psi = PG_+\phi = \phi$ we see that $\phi\in P(\DD(M,S))$.
This shows exactness at the second $\DD(M,S)$.

It remains to show that any $\phi\in \Cinfsc(M,S)$ with $P\phi=0$ is of the form $\phi=G\psi$ with $\psi\in \DD(M,S)$.
Using a cut-off function decompose $\phi$ as $\phi=\phi_+-\phi_-$ where $\supp(\phi_\pm) \subset J_\pm(K)$ where $K$ is a suitable compact subset of $M$.
Then $\psi:=P\phi_+=P\phi_-$ satisfies $\supp(\psi)\subset J_+(K)\cap J_-(K)$.
Thus $\psi\in \DD(M,S)$.
We check that $G_+\psi=\phi_+$.
Namely, for all $\chi\in\DD(M,S^*)$ we have by Lemma~\ref{l:GpmDsa}
$$
\int_M \<\chi,G_+P\phi_+\> \dV =
\int_M \<G_-^{*}\chi,P\phi_+\> \dV =
\int_M \<P^*G_-^{*}\chi,\phi_+\> \dV =
\int_M \<\chi,\phi_+\> \dV .
$$
The integration by parts in the second equality is justified because $\supp(\phi_+) \cap \supp(G^*_-\chi) \subset J_+(K)\cap J_-(\supp(\chi))$ is compact.
Similarly, one shows $G_-\psi=\phi_-$.
Now $G\psi = G_+\psi - G_-\psi = \phi_+ - \phi_- = \phi$ which concludes the proof. 
\end{proof}

In particular, given an object $(M,S,P)$ in $\GlobHypGreen$, the map $G$ induces an isomorphism from 
$$
\SYMPL(M,S,P)=\DD(M,S)/\ker(G) \xrightarrow{\cong} \ker(P)\cap \Cinfsc(M,S).
$$ 

\begin{rem}\label{rem:futurecompact}
Exactness at the first $\DD(M,S)$ in sequence \eqref{exactseq} says that there are no non-trivial smooth solutions of $P\phi=0$ with compact support.
Indeed, if $M$ is globally hyperbolic, more is true.

\emph{If $\phi\in\Cinf(M,S)$ solves $P\phi=0$ and $\supp(\phi)$ is future or past-compact, then $\phi=0$.}

Here a subset $A\subset M$ is called future-compact if $A\cap J_+(x)$ is compact for any $x\in M$.
Past-compactness is defined similarly.

\begin{proof}
Let $\phi\in\Cinf(M,S)$ solve $P\phi=0$ such that $\supp(\phi)$ is future-compact.
For any $\chi\in\DD(M,S^*)$ we have
$$
\int_M\<\chi,\phi\>\dV
=
\int_M\<P^*G^*_+\chi,\phi\>\dV
=
\int_M\<G^*_+\chi,P\phi\>\dV
=
0.
$$
This shows $\phi=0$.
The integration by parts is justified because $\supp(G^*_+\chi)\cap\supp(\phi) \subset J_+(\supp(\chi))\cap\supp(\phi)$ is compact, see \cite[Lemma~A.5.3]{BGP}.
\end{proof}
\end{rem}

\begin{rem}\label{rem:GreenUnique}
Let $M$ be a globally hyperbolic spacetime and $(M,S,P)$ an object in $\GlobHypGreen$.
{\em Then the Green's operators $G_+$ and $G_-$ are unique.}
Namely, if $G_+$ and $\tilde G_+$ are advanced Green's operators for $P$, then for any $\phi\in\DD(M,S)$ the section $\psi:=G_+\phi-\tilde G_+\phi$ has past-compact support and satisfies $P\psi=0$.
By the previous remark, we have $\psi=0$ which shows $G_+=\tilde G_+$.
\end{rem}

Now, let $(f,F)$ be a morphism between two objects $(M_1,S_1,P_1)$ and $(M_2,S_2,P_2)$ in the category $\GlobHypGreen$.
For $\phi\in\DD(M_1,S_1)$ consider the extension by zero $\ext(\phi)\in\DD(M_2,S_2)$ as in Lemma~\ref{l:extcommutesGreen}. 

\begin{lemma}\label{l:symplfunctor}
Given a morphism $(f,F)$ between two objects $(M_1,S_1,P_1)$ and $(M_2,S_2,P_2)$ in the category $\GlobHypGreen$, extension by zero induces a symplectic linear map $\SYMPL(f,F):\SYMPL(M_1,S_1,P_1)\rightarrow \SYMPL(M_2,S_2,P_2)$.

Moreover, 
\begin{equation}
\SYMPL(\id_M,\id_S)=\id_{\SYMPL(M,S,P)}
\label{eq:symplid} 
\end{equation}
and for any further morphism $(f',F'):(M_2,S_2,P_2)\rightarrow(M_3,S_3,P_3)$ one has
\begin{equation}
\SYMPL((f',F')\circ(f,F))=\SYMPL(f',F')\circ\SYMPL(f,F).
\label{eq:symplfunk}
\end{equation}
\end{lemma}

\begin{proof}
If $\phi=P_1\psi\in\ker(G_1)=P_1(\DD(M_1,S_1))$, then $\ext(\phi)=P_2(\ext(\psi))\in P_2(\DD(M_2,S_2))=\ker(G_2)$.
Hence $\ext$ induces a linear map 
$$
\SYMPL(f,F):\DD(M_1,S_1)/\ker(G_1)\to \DD(M_2,S_2)/\ker(G_2).
$$
Furthermore, applying Lemma~\ref{l:extcommutesGreen}, we have, for any $\phi,\psi\in\DD(M_1,S_1)$
$$
\int_{M_2}\< G_2(\ext(\phi)),\ext(\psi)\> \dV
=
\int_{M_1}\<\res\circ G_2\circ\ext(\phi),\psi\> \dV
=
\int_{M_1}\<G_1\phi,\psi\> \dV,
$$ 
hence $\SYMPL(f,F)$ is symplectic.
Equation \eqref{eq:symplid} is trivial and extending once or twice by $0$ amounts to the same, so (\ref{eq:symplfunk}) holds as well.
\end{proof}

\begin{rem}\label{rem:symplanders}
Under the isomorphism $\SYMPL(M,S,P) \to \ker(P) \cap \Cinfsc(M,S)$ induced by $G$, the extension by zero corresponds to an extension as a smooth solution of $P\phi=0$ with spacelike compact support.
This follows directly from Lemma~\ref{l:extcommutesGreen}.
In other words, for any morphism $(f,F)$ from $(M_1,S_1,P_1)$ to $(M_2,S_2,P_2)$ in $\GlobHypGreen$ we have the following commutative diagram:
$$
\xymatrix{
\SYMPL(M_1,S_1,P_1) \ar[rr]^{\SYMPL(f,F)} \ar[d]_\cong
&&  
\SYMPL(M_2,S_2,P_2) \ar[d]^\cong\\
\ker(P_1)\cap \Cinfsc(M_1,S_1)  \ar[rr]^{\mathrm{extension\, as}}_{\mathrm{a\, solution}}
&&  
\ker(P_2)\cap \Cinfsc(M_2,S_2) .
}
$$
\end{rem}

Let $\Sympl$ denote the category of real symplectic vector spaces with symplectic linear maps as morphisms.
Lemma~\ref{l:symplfunctor} says that we have constructed a covariant functor
\[\SYMPL:\GlobHypGreen\longrightarrow\Sympl.\]

In order to obtain an algebra-valued functor, we compose $\SYMPL$ with the functor $\CCR$ which associates to any symplectic vector space its Weyl algebra.
Here ``CCR'' stands for ``canonical commutation relations''.
This is a general algebraic construction which is independent of the context of Green-hyperbolic operators and which is carried out in Section~\ref{s:appendixCCR}. 
As a result, we obtain the functor 
\[\Abos := \CCR\circ\SYMPL:\GlobHypGreen\longrightarrow\CAlg,\]
where $\CAlg$ is the category whose objects are the unital \Cstar-algebras and whose morphisms are the injective unit-preserving \Cstar-morphisms.

In the remainder of this section we show that the functor $\CCR\circ\SYMPL$ is a bosonic locally covariant quantum field theory.
We call two subregions $M_1$ and $M_2$ of a spacetime $M$ \emph{causally disjoint} if and only if $J^M(M_1)\cap M_2=\emptyset$.
In other words, there are no causal curves joining $M_1$ and $M_2$.

\begin{thm}\label{thm:Abos}
The functor $\Abos:\GlobHypGreen\longrightarrow\CAlg$ is a bosonic locally covariant quantum field theory, i.e., the following axioms hold:
\begin{enumerate}[(i)]
\item\label{quantcaus}
\textbf{(Quantum causality)}
Let $(M_j,S_j,P_j)$ be objects in $\GlobHypGreen$, $j=1,2,3$, and $(f_j,F_j)$ morphisms from $(M_j,S_j,P_j)$ to $(M_3,S_3,P_3)$, $j=1,2$, such that $f_1(M_1)$ and $f_2(M_2)$ are causally disjoint regions in $M_3$.

Then the subalgebras $\Abos(f_1,F_1)(\Abos(M_1,S_1,P_1))$ and $\Abos(f_2,F_2)(\Abos(M_2,S_2,P_2))$ of\, $\Abos(M_3,S_3,P_3)$ commute. 
\item\label{timeslice}
\textbf{(Time slice axiom)}
Let $(M_j,S_j,P_j)$ be objects in $\GlobHypGreen$, $j=1,2$, and $(f,F)$ a morphism from $(M_1,S_1,P_1)$ to $(M_2,S_2,P_2)$ such that there is a Cauchy hypersurface $\Sigma\subset M_1$ for which $f(\Sigma)$ is a Cauchy hypersurface of $M_2$.
Then 
$$
\Abos(f,F):\Abos(M_1,S_1,P_1) \to \Abos(M_2,S_2,P_2)
$$ 
is an isomorphism.
\end{enumerate}
\end{thm}

\begin{proof}
We first show \eqref{quantcaus}.
For notational simplicity we assume without loss of generality that $f_j$ and $F_j$ are inclusions, $j=1,2$.
Let $\phi_j\in\DD(M_j,S_j)$.
Since $M_1$ and $M_2$ are causally disjoint, the sections $G\phi_1$ and $\phi_2$ have disjoint support, thus
$$
\omega([\phi_1],[\phi_2])=\int_M\< G\phi_1,\phi_2\> \dV = 0.
$$
Now relation \eqref{CCR4} in Definition~\ref{d:CCR} tells us
$$
w([\phi_1])\cdot w([\phi_2]) = w([\phi_1]+[\phi_2]) = w([\phi_2])\cdot w([\phi_1]) .
$$
Since $\Abos(f_1,F_1)(\Abos(M_1,S_1,P_1))$ is generated by elements of the form $w([\phi_1])$ and $\Abos(f_2,F_2)(\Abos(M_2,S_2,P_2))$ by elements of the form $w([\phi_2])$, the assertion follows.

In order to prove \eqref{timeslice} we show that $\SYMPL(f,F)$ is an isomorphism of symplectic vector spaces provided $f$ maps a Cauchy hypersurface of $M_1$ onto a Cauchy hypersurface of $M_2$.
Since symplectic linear maps are always injective, we only need to show surjectivity of $\SYMPL(f,F)$.
This is most easily seen by replacing $\SYMPL(M_j,S_j,P_j)$ by $\ker(P_j)\cap\Cinfsc(M_j,S_j)$ as in Remark~\ref{rem:symplanders}.
Again we assume without loss of generality that $f$ and $F$ are inclusions.

Let $\psi\in\Cinfsc(M_2,S_2)$ be a solution of $P_2\psi=0$.
Let $\phi$ be the restriction of $\psi$ to $M_1$.
Then $\phi$ solves $P_1\phi=0$ and has spacelike compact support in $M_1$ by Lemma~\ref{lem:spacecompact} below.
We will show that there is only one solution in $M_2$ with spacelike compact support extending $\phi$.
It will then follow that $\psi$ is the image of $\phi$ under the extension map corresponding to $\SYMPL(f,F)$ and surjectivity will be shown.

To prove uniqueness of the extension, we may, by linearity, assume that $\phi=0$.
Then $\psi_+$ defined by
$$
\psi_+(x) := \begin{cases}
              \psi(x), & \mbox{if $x\in J_+^{M_2}(\Sigma)$,}\\
               0,      & \mbox{otherwise,}
             \end{cases}
$$
is smooth since $\psi$ vanishes in an open neighborhood of $\Sigma$.
Now $\psi_+$ solves $P_2\psi_+=0$ and has past-compact support.
By Remark~\ref{rem:futurecompact}, $\psi_+\equiv 0$, i.e., $\psi$ vanishes on $J_+^{M_2}(\Sigma)$.
One shows similarly that $\psi$ vanishes on $J_-^{M_2}(\Sigma)$, hence $\psi=0$.
\end{proof}

\begin{lemma}\label{lem:spacecompact}
Let $M$ be a globally hyperbolic spacetime and let $M' \subset M$ be a causally compatible open subset which contains a Cauchy hypersurface of $M$.
Let $A\subset M$ be spacelike compact in $M$.

Then $A\cap M'$ is spacelike compact in $M'$.
\end{lemma}

\begin{proof}
Fix a common Cauchy hypersurface $\Sigma$ of $M'$ and $M$.
By assumption, there exists a compact subset $K\subset M$ with $A\subset
J^M(K)$.
Then $K':=J^M(K)\cap\Sigma$ is compact \cite[Cor.~A.5.4]{BGP} and contained in $M'$.

Moreover $A\subset J^M(K')$: let $p\in A$ and let $\gamma$ be a causal curve (in $M$) from $p$ to some $k\in K$.
Then $\gamma$ can be extended to an inextensible causal curve in $M$, which hence meets $\Sigma$ at some point $q$.
Because of $q\in\Sigma\cap J^M(k)\subset
K'$ one has $p\in J^M(K')$.

Therefore $A\cap M'\subset J^M(K')\cap M'=J^{M'}(K')$ because of the
causal compatibility of $M'$ in $M$.
The lemma is proved.
\end{proof}

The quantization process described in this subsection applies in particular to formally self-adjoint wave and Dirac-type operators.

\subsection{Fermionic quantization}

Next we construct a fermionic quantization.
For this we need a functorial construction of Hilbert spaces rather than symplectic vector spaces.
As we shall see this seems to be possible only under much more restrictive assumptions.
The underlying Lorentzian manifold $M$ is assumed to be a globally hyperbolic spacetime as before.
The vector bundle $S$ is assumed to be complex with Hermitian inner product $\<\cdot\,,\cdot\>$ which may be indefinite.
The formally self-adjoint Green-hyperbolic operator $P$ is assumed to be of first order.

\begin{definition}\label{d:Diractypdef}
A formally self-adjoint Green-hyperbolic operator $P$ of first order acting on sections of a complex vector bundle $S$ over a spacetime $M$ is of \emph{definite type} if and only if for any $x\in M$ and any future-directed timelike tangent vector $\mathfrak{n}\in T_xM$, the bilinear map 
$$
S_x\times S_x \to \C,
\quad\quad
(\phi,\psi)\mapsto\< i\sigma_P(\mathfrak{n}^\flat)\cdot\phi,\psi\> ,
$$
yields a positive definite Hermitian scalar product on $S_x$.
\end{definition}

\begin{ex}\label{ex:spinDiractypdef}
The classical Dirac operator $P$ from Example~\ref{ex:spinors} is, when defined with the correct sign, of definite type, see e.g. \cite[Sec.~1.1.5]{Baum} or \cite[Sec.~2]{BGM}.
\end{ex}

\begin{ex}\label{ex:twistedDiracnotdef}
If $E\to M$ is a semi-Riemannian or -Hermitian vector bundle endowed with a metric connection over a spin spacetime $M$, then the twisted Dirac operator from Example~\ref{ex:twistedDirac} is of definite type if and only if the metric on $E$ is positive definite.
This can be seen by evaluating the tensorized inner product on elements of the form $\sigma\otimes v$, where $v\in E_x$ is null.
\end{ex}

\begin{ex}\label{ex:Eulerntypdef}
The operator $P=i(d-\delta)$ on $S=\Lambda T^*M\otimes \C$ is of Dirac type but not of definite type.
This follows from Example~\ref{ex:twistedDiracnotdef} applied to Example~\ref{ex:Euler}, since the natural inner product on $\Sigma M$ is not positive definite.
An alternative elementary proof is the following: for any timelike tangent vector $\mathfrak{n}$ on $M$ and the corresponding covector $\mathfrak{n}^\flat$, one has 
$$
\< i\sigma_P(\mathfrak{n}^\flat)\mathfrak{n}^\flat,\mathfrak{n}^\flat\>
=
-\<\mathfrak{n}^\flat\wedge\mathfrak{n}^\flat-\mathfrak{n}\lrcorner\mathfrak{n}^\flat,\mathfrak{n}^\flat\>
=
\<\mathfrak{n},\mathfrak{n}\>\<1,\mathfrak{n}^\flat\>
=
0.
$$
\end{ex}

\begin{ex}\label{ex:RSnottypdef}
The Rarita-Schwinger operator defined in Section~\ref{ss:RaritaSchwinger} is not of definite type if the dimension of the manifolds is $m\geq 3$.
This can be seen as follows.
Fix a point $x\in M$ and a pointwise orthonormal basis $(e_j)_{1\leq j\leq m}$ of $T_xM$ with $e_1$ timelike.
The Lorentzian metric induces inner products on $\Sigma M$ and on $\Siggi$ which we denote by $\<\cdot\,,\cdot\>$.
Choose $\xi:=e_1^\flat\in T_x^*M$ and $\psi\in \Sigma^{3/2}_xM$.
Since $\sigma_\RS(\xi)$ is pointwise obtained as the orthogonal projection of $\sigma_{\mathcal{D}}(\xi)$ onto $\Sigma^{3/2}_xM$, one has
\be 
\<-i\sigma_\RS(\xi)\psi,\psi\>&=&\<(\id\otimes\xi^\sharp\cdot)\psi,\psi\>-\frac{2}{m}\underbrace{\sum_{\beta=1}^m\<e_\beta^*\otimes e_\beta\cdot\psi_1,\psi\>}_{=0}\\
&=&\sum_{\beta=1}^m\epsilon_\beta\<e_1\cdot\psi_\beta,\psi_\beta\>.
\ee
Choose, as in the proof of Lemma~\ref{l:RSnotdef}, a $\psi\in \Sigma^{3/2}_xM$ with $\psi_k=0$ for all $3\leq k\leq m$.
For such a $\psi$ the condition $\psi\in \Sigma^{3/2}_xM$ becomes $e_1\cdot\psi_1=e_2\cdot\psi_2$.
As in the proof of Lemma~\ref{l:RSnotdef} we obtain
\[\<-i\sigma_\RS(\xi)\psi,\psi\>=-\<e_1\cdot\psi_2,\psi_2\>+\<e_1\cdot\psi_2,\psi_2\>=0, \]
which shows that the Rarita-Schwinger operator cannot be of definite type.
\end{ex}

We define the category $\GlobHypDef$, whose objects are the triples $(M,S,P)$, where $M$ is a globally hyperbolic spacetime, $S$ is a complex vector bundle equipped with a complex inner product $\<\cdot\,,\cdot\>$, and $P$ is a formally self-adjoint Green-hyperbolic operator of definite type acting on sections of $S$.
The morphisms are the same as in the category $\GlobHypGreen$.

We construct a covariant functor from $\GlobHypDef$ to $\hilb$, where $\hilb$ denotes the category whose objects are complex pre-Hilbert spaces and whose morphisms are isometric linear embeddings.
As in Section~\ref{ss:CCRquant}, the underlying vector space is the space of classical solutions to the equation $P\phi=0$ with spacelike compact support.
We put
\[\SOL(M,S,P):=\ker(P)\cap \Cinfsc(M,S).\]
Here ``$\SOL$'' stands for classical solutions of the equation $P\phi=0$ with spacelike compact support.

\begin{lemma}\label{l:scalprodinvariant}
Let $(M,S,P)$ be an object in $\GlobHypDef$.
Let $\Sigma\subset M$ be a smooth spacelike Cauchy hypersurface with its future-oriented unit normal vector field $\mathfrak{n}$ and its induced volume element $\ds$.
Then
\begin{equation}
(\phi,\psi):=\int_\Sigma \< i\sigma_P(\mathfrak{n}^\flat)\cdot\phi_{|_\Sigma},\psi_{|_\Sigma}\>\ds,
\label{eq:ScalProd}
\end{equation}
yields a positive definite Hermitian scalar product on $\SOL(M,S,P)$ which does not depend on the choice of $\Sigma$.
\end{lemma}

\begin{proof}
First note that $\supp(\phi)\cap\Sigma$ is compact since $\supp(\phi)$ is spacelike compact, so that the integral is well-defined.
We have to show that it does not depend on the choice of Cauchy hypersurface.
Let $\Sigma'$ be any other smooth spacelike Cauchy hypersurface.
Assume first that $\Sigma$ and $\Sigma'$ are disjoint and let $\Omega$ be the domain enclosed by $\Sigma$ and $\Sigma'$ in $M$.
Its boundary is $\partial\Omega=\Sigma\cup\Sigma'$.
Without loss of generality, one may assume that $\Sigma'\subset J_+^M(\Sigma)$.
By the Green's formula \cite[p.~160, Prop.~9.1]{T1} we have for all $\phi,\psi \in \Cinfsc(M,S)$, 
\begin{equation}\label{eq:Green1storder}
\int_\Omega \left(\< P\phi,\psi\>-\< \phi,P\psi\>\right) \dV
=
\int_{\Sigma'}\<\sigma_P(\mathfrak{n}^\flat)\phi,\psi\> \ds
-\int_{\Sigma}\<\sigma_P(\mathfrak{n}^\flat)\phi,\psi\> \ds.
\end{equation}
For $\phi,\psi\in\SOL(M,S,P)$ we have $P\phi=P\psi=0$ and thus
\begin{align*}
0 
&=
\int_{\Sigma}\<\sigma_P(\mathfrak{n}^\flat)\phi,\psi\>\ds 
- \int_{\Sigma'}\<\sigma_P(\mathfrak{n}^\flat)\phi,\psi\>\ds .
\end{align*}
This shows the result in the case $\Sigma\cap\Sigma'=\emptyset$.

If $\Sigma\cap\Sigma'\neq\emptyset$ consider the subset $I_-^M(\Sigma)\cap I_-^M(\Sigma')$ of $M$ where, as usual, $I_+^M(\Sigma)$ and $I_-^M(\Sigma)$ denote the chronological future and past of the subset $\Sigma$ in $M$, respectively.
This subset is nonempty, open, and globally hyperbolic.
This follows e.g.\ from \cite[Lemma A.5.8]{BGP}.
Hence it admits a smooth spacelike Cauchy hypersurface $\Sigma''$ by Theorem~\ref{tBernalSanchez}.
By construction, $\Sigma''$ meets neither $\Sigma$ nor $\Sigma'$ and it can be easily checked that $\Sigma''$ is also a Cauchy hypersurface of $M$. 
The result follows from the argument above being applied first to the pair $(\Sigma,\Sigma'')$ and then to the pair $(\Sigma'',\Sigma')$.
\end{proof}

\begin{rem}
If one drops the assumption that $P$ be of definite type, then the above sesquilinear form $(\cdot\,,\cdot)$ on $\ker(P)\cap\Cinfsc(M,S)$ still does not depend on the choice of $\Sigma$, however it need no longer be positive definite and can even be degenerate.
Pick for instance the spin Dirac operator $D_g$ associated to the underlying Lorentzian metric $g$ on a spin spacetime $M$ (see Example~\ref{ex:spinors}) and, keeping the spinor bundle $\Sigma_gM$ associated to $g$, change the metric on $M$ so that the new metric $g'$ has larger future and past cones at each point. Note that this implies that any globally hyperbolic subregion of $(M,g')$ is also globally hyperbolic in $(M,g)$.
Then, denoting by $D_g^*$ the formal adjoint of $D_g$ with respect to the metric $g'$, the operator $\left(\begin{array}{cc}0& D_g\\D_g^*&0\end{array}\right)$ on $\Sigma_gM\oplus\Sigma_gM$ remains Green-hyperbolic but it fails to be of definite type, since there exist timelike vectors for $g'$ which are lightlike for $g$.
Hence the principal symbol of the operator becomes non-invertible and the bilinear form in \eqref{eq:ScalProd} becomes degenerate for these $g'$-timelike covectors.
\end{rem}

For any object $(M,S,P)$ in $\GlobHypDef$ we will from now on equip $\SOL(M,S,P)$ with the Hermitian scalar product in \eqref{eq:ScalProd} and thus turn $\SOL(M,S,P)$ into a pre-Hilbert space.

Given a morphism $(f,F)$ from $(M_1,S_1,P_1)$ to $(M_2,S_2,P_2)$ in $\GlobHypDef$, then this is also a morphism in $\GlobHypGreen$ and hence induces a homomorphism $\SYMPL(f,F):\SYMPL(M_1,S_1,P_1) \to \SYMPL(M_2,S_2,P_2)$.
As explained in Remark~\ref{rem:symplanders}, there is a corresponding extension homomorphism $\SOL(f,F):\SOL(M_1,S_1,P_1) \to \SOL(M_2,S_2,P_2)$.
In other words, $\SOL(f,F)$ is defined such that the diagram
\begin{equation}
\xymatrix{
\SYMPL(M_1,S_1,P_1) \ar[rr]^{\SYMPL(f,F)} \ar[d]_\cong
&&  
\SYMPL(M_2,S_2,P_2) \ar[d]^\cong\\
\SOL(M_1,S_1,P_1)  \ar[rr]^{\SOL(f,F)}
&&  
\SOL(M_2,S_2,P_2)
}
\label{eq:symplhilb}
\end{equation}
commutes.
The vertical arrows are the vector space isomorphisms induced be the Green's propagators $G_1$ and $G_2$, respectively.

\begin{lemma}\label{lem:hilb}
The vector space homomorphism $\SOL(f,F):\SOL(M_1,S_1,P_1) \to \SOL(M_2,S_2,P_2)$ preserves the scalar products, i.e., it is an isometric linear embedding of pre-Hilbert spaces.
\end{lemma}

\begin{proof}
Without loss of generality we assume that $f$ and $F$ are inclusions.
Let $\Sigma_1$ be a spacelike Cauchy hypersurface of $M_1$.
Let $\phi_1,\psi_1\in\Cinfsc(M_1,S_1)$.
Denote the extension of $\phi_1$ by $\phi_2:=\SOL(f,F)(\phi_1)$ and similarly for $\psi_1$.

Let $K_1\subset M_1$ be a compact subset such that $\supp(\phi_2)\subset J^{M_2}(K_1)$ and $\supp(\psi_2)\subset J^{M_2}(K_1)$.
We choose a compact submanifold $K\subset \Sigma_1$ with boundary such that $J^{M_1}(K_1) \cap\Sigma_1\subset K$.
Since $\Sigma_1$ is a Cauchy hypersurface in $M_1$, $J^{M_1}(K_1) \subset J^{M_1}(J^{M_1}(K_1)\cap\Sigma_1) \subset J^{M_1}(K)$.

By Theorem~\ref{tBernalSanchez3} there is a spacelike Cauchy hypersurface $\Sigma_2 \subset M_2$ containing $K$.
Since $\Sigma_i$ is a Cauchy hypersurface of $M_i$ (where $i=1,2$), it is met by every inextensible causal curve \cite[Lemma 14.29]{ONeill}.
Moreover, by definition of a Cauchy hypersurface, $\Sigma_i$ is achronal in $M_i$.
Since it is also spacelike, $\Sigma_i$ is even acausal \cite[Lemma 14.42]{ONeill}. In particular, it is met \emph{exactly once} by every inextensible causal curve in $M_i$.

This implies $J^{M_2}(K_1)\subset J^{M_2}(K)$: 
namely, pick $p\in J^{M_2}(K_1)$ and a causal curve $\gamma$ in $M_2$ from $p$ to some $k_1\in K_1$.
Extend $\gamma$ to an inextensible causal curve $\overline{\gamma}$ in $M_2$.
Then $\overline{\gamma}$ meets $\Sigma_2$ at some point $q_2$, because $\Sigma_2$ is a Cauchy hypersurface in $M_2$.
But $\overline{\gamma}\cap M_1$ is also an inextensible causal curve in $M_1$, hence it intersects $\Sigma_1$ at a point $q_1$, which must lie in $K$ by definition of $K$.
Because of $K\subset\Sigma_2$ and the uniqueness of the intersection point, one has $q_1=q_2$.
In particular, $p\in J^{M_2}(K)$.

\begin{center}
\scalebox{0.6} 
{
\begin{pspicture}(0,-4.86)(20.7,5.86)
\pscustom[linewidth=0.04,fillstyle=solid,fillcolor=lightgray]{
\psbezier(4.75,-0.34)(5.88,-2.36)(13.752771,-3.48)(16.22,-0.34)
\psbezier[liftpen=1](16.22,-0.34)(14.62,2.76)(6.4,2.3351269)(4.75,-0.34)
}
\psbezier[linewidth=0.04](0.02,-0.3)(4.24,-5.84)(19.04,-3.78)(20.68,-0.4)
\psbezier[linewidth=0.04](0.02,-0.3)(3.56,4.0452337)(17.54,5.84)(20.68,-0.4)

\psbezier[linewidth=0.04,fillstyle=solid,fillcolor=gray](9.7,0.9)(9.689865,1.8999486)(10.037922,1.288475)(11.02,1.1)(12.002078,0.911525)(12.1638975,1.5660923)(11.84,0.62)(11.516103,-0.3260923)(9.7101345,-0.099948645)(9.7,0.9)

\psbezier[linewidth=0.07,tbarsize=0.07055555cm 5.0]{|-|}(6.86,-0.12)(7.755464,-0.17354254)(13.935233,-0.34)(14.94,-0.32043934)

\psline[linewidth=0.02cm](11.74,0.44)(14.58,3.38)
\psline[linewidth=0.02cm](9.72,0.62)(7.74,3.3)
\psline[linewidth=0.02cm](11.96,1.16)(16.24,-2.9)
\psline[linewidth=0.02cm](9.76,1.44)(5.54,-3.28)

\psbezier[linewidth=0.04](4.75,-0.34)(5.72,-0.44)(5.9,-0.12)(7.1,-0.12)
\psbezier[linewidth=0.04](16.22,-0.34)(14.84,-0.3)(15.76,-0.34)(14.58,-0.32)
\psbezier[linewidth=0.04](0.02,-0.3)(1.62,0.52)(5.66,-0.2)(7.42,-0.14)
\psbezier[linewidth=0.04](20.68,-0.4)(18.58,0.64)(15.66,-0.46)(14.24,-0.3)

\rput(16.601406,1.725){$M_2$}
\rput(10.711407,0.685){\psframebox*[framearc=.3]{$K_1$}}
\rput(7.621406,0.845){\psframebox*[framearc=.3]{$M_1$}}
\rput(11.241406,2.805){$J^{M_2}(K_1)$}
\rput(2.7714062,0.505){$\Sigma_2$}
\rput(5.891406,-0.595){\psframebox*[framearc=.3]{$\Sigma_1$}}
\rput(12.651406,-0.315){\psframebox*[framearc=.3]{$K$}}
\end{pspicture} 
}
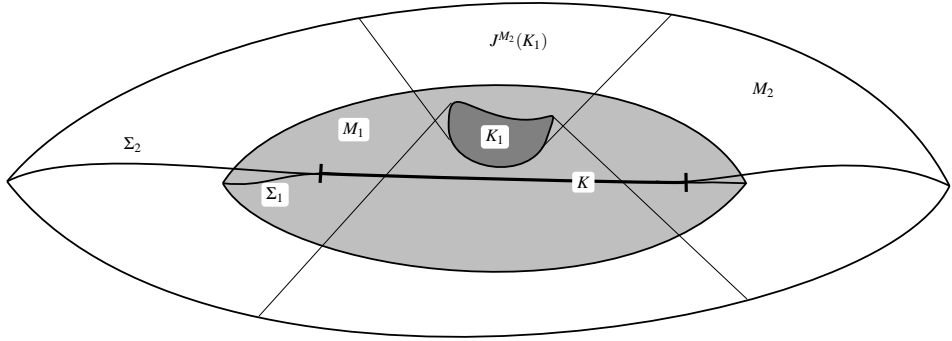
{\em Fig.~1}
\end{center}

We conclude $\supp(\phi_2)\subset J^{M_2}(K)$.
Since $K\subset \Sigma_2$, we have $\supp(\phi_2) \cap\Sigma_2\subset J^{M_2}(K)\cap\Sigma_2$ and $J^{M_2}(K)\cap\Sigma_2=K$ using the acausality of $\Sigma_2$.
This shows $\supp(\phi_2) \cap\Sigma_2 = \supp(\phi_1) \cap\Sigma_1$ and similarly for $\psi_2$.
Now we get
$$
(\phi_2,\psi_2)
=
\int_{\Sigma_2} \< i\sigma_{P_2}(\mathfrak{n}^\flat)\cdot\phi_2,\psi_2\>\ds
=
\int_{\Sigma_1} \< i\sigma_{P_1}(\mathfrak{n}^\flat)\cdot\phi_1,\psi_1\>\ds
=
(\phi_1,\psi_1)
$$
and the lemma is proved.
\end{proof}

The functoriality of $\SYMPL$ and diagram~\eqref{eq:symplhilb} show that $\SOL$ is a functor from $\GlobHypDef$ to $\hilb$, the category of complex pre-Hilbert spaces with isometric linear embeddings.
Composing with the functor $\mathrm{CAR}$ (see Section~\ref{s:appendixCAR}), we obtain the covariant functor 
\[
\Aferm:=\mathrm{CAR}\circ\SOL:\GlobHypDef\longrightarrow\CAlg.
\]
The fermionic algebras $\Aferm(M,S,P)$ are actually $\Z_2$-graded algebras, see Proposition~\ref{pCAR}~\eqref{pCAR:4}.

\begin{thm}\label{thm:Aferm}
The functor $\Aferm:\GlobHypDef\longrightarrow\CAlg$ is a fermionic locally covariant quantum field theory, i.e., the following axioms hold:
\begin{enumerate}[(i)]
\item\label{quantcaus2}
\textbf{(Quantum causality)}
Let $(M_j,S_j,P_j)$ be objects in $\GlobHypDef$, $j=1,2,3$, and $(f_j,F_j)$ morphisms from $(M_j,S_j,P_j)$ to $(M_3,S_3,P_3)$, $j=1,2$, such that $f_1(M_1)$ and $f_2(M_2)$ are causally disjoint regions in $M_3$.\\
Then the subalgebras $\Aferm(f_1,F_1)(\Aferm(M_1,S_1,P_1))$ and $\Aferm(f_2,F_2)(\Aferm(M_2,S_2,P_2))$ of $\Aferm(M_3,S_3,P_3)$ super-commute\footnote{This means that the odd parts of the algebras anti-commute while the even parts commute with everything.}. 
\item\label{timeslice2}
\textbf{(Time slice axiom)}
Let $(M_j,S_j,P_j)$ be objects in $\GlobHypDef$, $j=1,2$, and $(f,F)$ a morphism from $(M_1,S_1,P_1)$ to $(M_2,S_2,P_2)$ such that there is a Cauchy hypersurface $\Sigma\subset M_1$ for which $f(\Sigma)$ is a Cauchy hypersurface of $M_2$.
Then 
$$
\Aferm(f,F):\Aferm(M_1,S_1,P_1) \to \Aferm(M_2,S_2,P_2)
$$ 
is an isomorphism.
\end{enumerate}
\end{thm}

\begin{proof}
To show \eqref{quantcaus2}, we assume without loss of generality that $f_j$ and $F_j$ are inclusions.
Let $\phi_1\in\SOL(M_1,S_1,P_1)$ and $\psi_1\in\SOL(M_2,S_2,P_2)$.
Denote the extensions to $M_3$ by $\phi_2:=\SOL(f_1,F_1)(\phi_1)$ and $\psi_2:=\SOL(f_2,F_2)(\psi_1)$.
Choose a compact submanifold $K_1$ (with boundary) in a spacelike Cauchy hypersurface $\Sigma_1$ of $M_1$ such that $\supp(\phi_1)\cap\Sigma_1 \subset K_1$ and similarly $K_2$ for $\psi_1$.
Since $M_1$ and $M_2$ are causally disjoint, $K_1 \cup K_2$ is acausal.
Hence, by Theorem~\ref{tBernalSanchez3}, there exists a Cauchy hypersurface $\Sigma_3$ of $M_3$ containing $K_1$ and $K_2$.
As in the proof of Lemma~\ref{lem:hilb} one sees that $\supp(\phi_2)\cap\Sigma_3=\supp(\phi_1)\cap\Sigma_1$ and similarly for $\psi_2$.
Thus, when restricted to $\Sigma_3$, $\phi_2$ and $\psi_2$ have disjoint support.
Hence $(\phi_2,\psi_2)=0$.
This shows that the subspaces $\SOL(f_1,F_1)(\SOL(M_1,S_1,P_1))$ and $\SOL(f_2,F_2)(\SOL(M_2,S_2,P_2))$ of $\SOL(M_3,S_3,P_3)$ are perpendicular.
Definition~\ref{def-CAR} shows that the corresponding CAR-algebras must super-commute.

To see \eqref{timeslice2} we recall that $(f,F)$ is also a morphism in $\GlobHypGreen$ and that we know from Theorem~\ref{thm:Abos} that $\SYMPL(f,F)$ is an isomorphism.
From diagram~\eqref{eq:symplhilb} we see that $\SOL(f,F)$ is an isomorphism.
Hence $\Aferm(f,F)$ is also an isomorphism.
\end{proof}

\begin{rem}
Since causally disjoint regions should lead to commuting observables also in the fermionic case, one usually considers only the even part $\Aferm^\mathrm{even}(M,S,P)$ (or a subalgebra thereof) as the observable algebra while the full algebra $\Aferm(M,S,P)$ is called the {\em field algebra}.
\end{rem}

There is a slightly different description of the functor $\Aferm$.
Let $\rhilb$ denote the category whose objects are the real pre-Hilbert spaces and whose morphisms are the isometric linear embeddings.
We have the functor $\forget:\hilb\to\rhilb$ which associates to each complex pre-Hilbert space $(V,(\cdot\,,\cdot))$ its underlying real pre-Hilbert space $(V,\Re(\cdot\,,\cdot))$.
By Remark~\ref{rem:RversusC},
$$
\Aferm=\mathrm{CAR}_\mathrm{sd}\circ\forget\circ\SOL .
$$
Since the self-dual CAR-algebra of a real pre-Hilbert space is the Clifford algebra of its complexification and since for any complex pre-Hilbert space $V$ we have $$
\forget(V) \otimes_\R \C = V \oplus V^*, 
$$
$\Aferm(M,S,P)$ is also the Clifford algebra of $\SOL(M,S,P) \oplus \SOL(M,S,P)^* = \SOL(M,S\oplus S^*,P\oplus P^*)$.
This is the way this functor is often described in the physics literature, see e.g.\ \cite[p.~115f]{St}.

Self-dual CAR-representations are more natural for real fields.
Let $M$ be globally hyperbolic and let $S\to M$ be a {\em real} vector bundle equipped with a real inner product $\<\cdot\,,\cdot\>$.
A formally skew-adjoint\footnote{instead of self-adjoint!} differential operator $P$ acting on sections of $S$ is called of {\em definite type} if and only if for any $x\in M$ and any future-directed timelike tangent vector $\mathfrak{n}\in T_xM$, the bilinear map 
$$
S_x\times S_x \to \R,
\quad\quad
(\phi,\psi)\mapsto\< \sigma_P(\mathfrak{n}^\flat)\cdot\phi,\psi\> ,
$$
yields a positive definite Euclidean scalar product on $S_x$.
An example is given by the real Dirac operator 
$$
D:=\sum_{j=1}^m\epsilon_j e_j\cdot\na_{e_j}
$$
acting on sections of the real spinor bundle $\Sigma^\R M$.

Given a smooth spacelike Cauchy hypersurface $\Sigma\subset M$ with future-directed timelike unit normal field $\mathfrak{n}$, we define a scalar product on $\SOL(M,S,P)=\ker(P)\cap \Cinfsc(M,S,P)$ by 
$$
(\phi,\psi):=\int_\Sigma \< \sigma_P(\mathfrak{n}^\flat)\cdot\phi_{|_\Sigma},\psi_{|_\Sigma}\>\ds .
$$
With essentially the same proofs as before, one sees that this scalar product does not depend on the choice of Cauchy hypersurface $\Sigma$ and that a morphism $(f,F):(M_1,S_1,P_1) \to (M_2,S_2,P_2)$ gives rise to an extension operator $\SOL(f,F):\SOL(M_1,S_1,P_1) \to \SOL(M_2,S_2,P_2)$ preserving the scalar product.
We have constructed a functor 
$$
\SOL:\GHSD\longrightarrow\rhilb
$$
where $\GHSD$ denotes the category whose objects are triples $(M,S,P)$ with $M$ globally hyperbolic, $S\to M$ a real vector bundle with real inner product and $P$ a formally skew-adjoint, Green-hyperbolic differential operator of definite type acting on sections of $S$.
The morphisms are the same as before.

Now the functor 
$$
\Afermsd := \CAR_\mathrm{sd} \circ \SOL : \GHSD \longrightarrow \CAlg
$$
is a locally covariant quantum field theory in the sense that Theorem~\ref{thm:Aferm} holds with $\Aferm$ replaced by $\Afermsd$.

\section{States and quantum fields}

In order to produce numbers out of our quantum field theory that can be compared to experiments, we need states, in addition to observables.
We briefly recall the relation between states and representations via the GNS-construction.
Then we show how the choice of a state gives rise to quantum fields and $n$-point functions.

\subsection{States and representations}
Recall that a \emph{state} on a unital \Cstar-algebra $A$ is a linear functional $\tau: A\to \C$ such that 
\begin{enumerate}[(i)]
\item
$\tau$ is positive, i.e., $\tau(a^*a)\ge 0$ for all $a\in A$;
\item 
$\tau$ is normed, i.e., $\tau(1)=1$.
\end{enumerate}
One checks that for any state the sesquilinear form $A\times A \to \C$, $(a,b) \mapsto \tau(b^*a)$, is a positive semi-definite Hermitian product and $|\tau(a)|\leq \|a\|$ for all $a\in A$.
In particular, $\tau$ is continuous.

Any state induces a representation of $A$.
Namely, the sesquilinear form $\tau(b^*a)$ induces a scalar product $\so{\cdot}{\cdot}$ on $A/\{a\in A\mid \tau(a^*a)=0\}$.
The Hilbert space completion of $A/\{a\in A\mid \tau(a^*a)=0\}$ is denoted by $\H_\tau$.
The action of $A$ on $\H_\tau$ is induced by the multiplication in $A$,
$$
\pi_\tau(a)[b]_\tau := [ab]_\tau,
$$
where $[a]_\tau$ denotes the residue class of $a\in A$ in $A/\{a\in A\mid \tau(a^*a)=0\}$.
This representation is known as the \emph{GNS-representation} induced by $\tau$.
The residue class $\Omega_\tau:=[1]_\tau\in \H_\tau$ is called the \emph{vacuum vector}.
By construction, it is a cyclic vector, i.e., the orbit $\pi_\tau(A)\cdot\Omega_\tau = A/\{a\in A\mid \tau(a^*a)=0\}$ is dense in $\H_\tau$.

The GNS-representation together with the vacuum vector allows to reconstruct the state since
\begin{equation}
\tau(a) = \tau(1^*a1) = \so{\pi_\tau(a)\Omega_\tau}{\Omega_\tau} .
\label{eq:tau}
\end{equation}
If we look at the vector state $\tilde\tau:\LL(\H_\tau) \to \C$, $\tilde\tau(\tilde a) = \so{\tilde a\,\Omega_\tau}{\Omega_\tau}$, on the \Cstar-algebra $\LL(\H_\tau)$ of bounded linear operators on $\H_\tau$, then \eqref{eq:tau} says that the diagram
$$
\xymatrix{
A \ar[rr]^{\pi_\tau} \ar[rd]_\tau && \LL(\H_\tau) \ar[ld]^{\tilde\tau}\\
& \C &
}
$$
commutes.
One checks that $\|\pi_\tau\|\le 1$, see \cite[p.~20]{BB}. 
In particular, $\pi_\tau: A \to \LL(\H_\tau)$ is continuous.

See e.g.\ \cite[Sec.~1.4]{BB} or \cite[Sec.~2.3]{BR} for details on states and representations of \Cstar-algebras.

\subsection{Bosonic quantum field}
Now let $(M,S,P)$ be an object in $\GlobHypGreen$ and $\tau$ a state on the corresponding bosonic algebra $\Abos(M,S,P)$.
Intuitively, the quantum field should be an operator-valued distribution $\Phi$ on $M$ such that
$$
e^{i\Phi(f)} = w([f])
$$
for all test sections $f\in\DD(M,S)$.
Here $[f]$ denotes the residue class in $\SYMPL(M,S,P)=\DD(M,S)/\ker G$ and $w:\SYMPL(M,S,P) \to \Abos(M,S,P)$ is as in Definition~\ref{d:CCR}.
This suggests the definition
$$
\Phi(f) := -i\ddt w(t[f]).
$$
The problem is that $w$ is highly discontinuous so that this derivative does not make sense.
This is where states and representations come into the play.
We call a state $\tau$ on $\Abos(M,S,P)$ {\em regular} if for each $f\in\DD(M,S)$ and each $h\in\H_\tau$ the map $t\mapsto \pi_\tau(w(t[f]))h$ is continuous.
Then $t\mapsto \pi_\tau(w(t[f]))$ is a strongly continuous one-parameter unitary group for any $f\in\DD(M,S)$ because
$$
\pi_\tau(w((t+s)[f])) 
= \pi_\tau(e^{i\omega(t[f],s[f])/2}w(t[f])w(s[f]))
= \pi_\tau(w(t[f]))\pi_\tau(w(s[f])) .
$$
Here we used Definition~\ref{d:CCR}~\eqref{CCR4} and the fact that $\omega$ is skew-symmetric so that $\omega(t[f],s[f])=0$.
By Stone's theorem \cite[Thm.~VIII.8]{RS1} this one-parameter group has a unique infinitesimal generator, i.e., a self-adjoint, generally unbounded operator $\Phi_\tau(f)$ on $\H_\tau$ such that 
$$
e^{it\Phi_\tau(f)} = \pi_\tau(w(t[f])) .
$$
For all $h$ in the domain of $\Phi_\tau(f)$ we have
$$
\Phi_\tau(f)h = -i\ddt \pi_\tau(w(t[f]))h .
$$
We call the operator-valued map $f \mapsto \Phi_\tau(f)$ the {\em quantum field}
corresponding to $\tau$.

\begin{definition}\label{def:strongreg}
A regular state $\tau$ on $\Abos(M,S,P)$ is called {\em strongly regular} if
\begin{enumerate}[(i)]
\item there is a dense subspace $\D_\tau\subset \H_\tau$ contained in the domain of $\Phi_\tau(f)$ for any $f\in\DD(M,S)$;
\item $\Phi_\tau(f)(\D_\tau)\subset \D_\tau$ for any $f\in\DD(M,S)$;
\item the map $\DD(M,S) \to \H_\tau$, $f \mapsto \Phi_\tau(f)h$, is continuous for every fixed $h\in\D_\tau$.
\end{enumerate}
\end{definition}

For a strongly regular state $\tau$ we have for all $f,g\in\DD(M,S)$, $\alpha,\beta\in\R$ and $h\in\D_\tau$:
\begin{align*}
\Phi_\tau(\alpha f + \beta g)h
&=
-i\ddt \pi_\tau(w(t[\alpha f +\beta g]))h \\
&=
-i\ddt\left\{ e^{i\alpha\beta t^2\omega([f],[g])/2}\pi_\tau(w(\alpha t[f]))\pi_\tau(w(\beta t[g]))h\right\} \\
&=
-i\ddt \pi_\tau(w(\alpha t[f]))h - i\ddt\pi_\tau(w(\beta t[g]))h \\
&=
\alpha \Phi_\tau(f)h + \beta\Phi_\tau(g)h.
\end{align*}
Hence $\Phi_\tau(f)$ depends linearly on $f$.
The quantum field $\Phi_\tau$ is therefore a distribution on $M$ with values in self-adjoint operators on $\H_\tau$.

The {\em $n$-point functions}  are defined by
\begin{align*}
\tau_n(f_1,\ldots,f_n)
&:=
\so{\Phi_\tau(f_1)\cdots\Phi_\tau(f_n)\Omega_\tau}{\Omega_\tau} \\
&=
\tilde\tau\left(\Phi_\tau(f_1)\cdots\Phi_\tau(f_n)\right)\\
&=
\tilde\tau\left(\left(-i\ddT{1} \pi_\tau(w(t_1[f_1]))\right)\cdots \left(-i\ddT{n} \pi_\tau(w(t_n[f_n]))\right)\right) \\
&=
(-i)^n \DDT \tilde\tau\left(\pi_\tau(w(t_1[f_1]))\cdots \pi_\tau(w(t_n[f_n]))\right) \\
&=
(-i)^n \DDT \tilde\tau\left(\pi_\tau(w(t_1[f_1])\cdots w(t_n[f_n]))\right) \\
&=
(-i)^n \DDT \tau\left(w(t_1[f_1])\cdots w(t_n[f_n])\right) .
\end{align*}

For a strongly regular state $\tau$ the $n$-point functions are continuous separately in each factor.
By the Schwartz kernel theorem \cite[Thm.~5.2.1]{Hoerm1} the $n$-point function $\tau_n$ extends uniquely to a distribution on $M\times\cdots\times M$ ($n$ times) in the following sense:
Let $S^*\boxtimes \cdots \boxtimes S^*$ be the bundle over $M\times\cdots\times M$ whose fiber over $(x_1,\ldots,x_n)$ is given by $S^*_{x_1}\otimes \cdots \otimes S^*_{x_n}$.
Then there is a unique distribution on $M\times\cdots\times M$ in the bundle $S^*\boxtimes \cdots \boxtimes S^*$, again denoted $\tau_n$, such that for all $f_j\in\DD(M,S)$,
$$
\tau_n(f_1,\ldots,f_n) = \tau_n(f_1\otimes \cdots  \otimes f_n)
$$
where $(f_1\otimes \cdots  \otimes f_n)(x_1,\ldots,x_n) := f_1(x_1)\otimes\cdots\otimes f_n(x_n)$.

\begin{thm}\label{thm:QFbosonic}
Let $(M,S,P)$ be an object in $\GlobHypGreen$ and $\tau$ a strongly regular state on the corresponding bosonic algebra $\Abos(M,S,P)$.
Then
\begin{enumerate}[(i)]
\item\label{bos:PPsiO}
$P\Phi_\tau=0$ and $P\tau_n(f_1,\ldots,f_{j-1},\cdot,f_{j+1},\ldots,f_n)=0$ hold in the distributional sense where $f_k\in\DD(M,S)$, $k\neq j$, are fixed;
\item\label{bos:CCRPsi}
the quantum field satisfies the canonical commutation relations, i.e., 
$$
[\Phi_\tau(f),\Phi_\tau(g)]h = i\int_M \<Gf,g\>\dV \cdot h
$$ 
for all $f,g\in\DD(M,S)$ and $h\in\D_\tau$;
\item\label{bos:CCRnpoint}
the $n$-point functions satisfy the canonical commutation relations, i.e., \begin{align*}
&\tau_{n+2}(f_1,\ldots,f_{j-1},f_j,f_{j+1},\ldots,f_{n+2})\\
& - \tau_{n+2}(f_1,\ldots,f_{j-1},f_{j+1},f_j,f_{j+2},\ldots,f_{n+2}) \\
=\,& i\int_M\<Gf_j,f_{j+1}\>\dV \cdot \tau_n(f_1,\ldots,f_{j-1},f_{j+2},\ldots,f_{n+2})
\end{align*}
for all $f_1,\ldots,f_{n+2}\in \DD(M,S)$.
\end{enumerate}
\end{thm}

\begin{proof}
Since $P$ is formally self-adjoint and $GPf=0$ for any $f\in\DD(M,S)$, we have for any $h\in\D_\tau$:
$$
(P\Phi_\tau)(f)h = \Phi_\tau(Pf)h = -i\ddt \pi_\tau (w(t\underbrace{[Pf]}_{=0}))h = -i \ddt h = 0.
$$
This shows $P\Phi_\tau=0$.
The result for the $n$-point functions follows and \eqref{bos:PPsiO} is proved.

To show \eqref{bos:CCRPsi} we observe that by Definition~\ref{d:CCR}~\eqref{CCR4} we have on the one hand
$$
w([f+g]) = e^{i\omega([f],[g])/2}w([f])w([g])
$$
and on the other hand
$$
w([f+g]) = e^{i\omega([g],[f])/2}w([g])w([f]),
$$
hence
$$
w([f])w([g]) = e^{-i\omega([f],[g])}w([g])w([f]) .
$$
Thus
\begin{align*}
\Phi_\tau(f)\Phi_\tau(g)h
&=
-\dtt \pi_\tau(w(t[f])w(s[g]))h \\
&=
-\dtt \pi_\tau(e^{-i\omega(t[f],s[g])}w(s[g])w(t[f]))h \\
&=
-\dtt \left\{e^{-i\omega(t[f],s[g])}\cdot\pi_\tau(w(s[g])w(t[f]))h\right\} \\
&=
i\omega([f],[g])h + \Phi_\tau(g)\Phi_\tau(f)h \\
&=
i\int_M\<Gf,g\>\dV\cdot h  + \Phi_\tau(g)\Phi_\tau(f)h.
\end{align*}
This shows \eqref{bos:CCRPsi}.
Assertion \eqref{bos:CCRnpoint} follows from \eqref{bos:CCRPsi}.
\end{proof}

\begin{rem}
As a consequence of the canonical commutation relations we get
$$
[\Phi_\tau(f),\Phi_\tau(g)]=0
$$
if the supports of $f$ and $g$ are causally disjoint, i.e., if there is no causal curve from $\supp(f)$ to $\supp(g)$.
The reason is that in this case the supports of $Gf$ and $g$ are disjoint.
A similar remark holds for the $n$-point functions.
\end{rem}

\begin{rem}
In the physics literature one also finds the statement $\Phi(\overline f) = \Phi(f)^*$.
This simply expresses the fact that we are dealing with a theory over the reals.
We have encoded this by considering real vector bundles $S$, see Definition~\ref{def:GHG}, and the fact that $\Phi_\tau(f)$ is always self-adjoint.
\end{rem}

\subsection{Fermionic quantum fields}

Let $(M,S,P)$ be an object in $\GlobHypDef$ and let $\tau$ be a state on the fermionic algebra $\Aferm(M,S,P)$.
For $f\in\DD(M,S)$ we put
\be
\Phi_\tau(f) &:=&- \pi_\tau(\a(Gf)^*),\\
\Phi_\tau^+(f)&:=&\pi_\tau(\a(Gf)),
\ee
where $\a$ is as in Definition~\ref{def-CAR} (compare \cite[Sec.~III.B, p.~141]{Dimock82}).
Since $\pi_\tau$, $\a$, and $G$ are sequentially continuous (for $G$ see \cite[Prop.~3.4.8]{BGP}), so are $\Phi_\tau$ and $\Phi_\tau^+$.
In contrast to the bosonic case, no regularity assumption on $\tau$ is needed.
Hence $\Phi_\tau$ and $\Phi_\tau^+$ are distributions on $M$ with values in the space of bounded operators on $\H_\tau$.
Note that $\Phi_\tau$ is linear while $\Phi_\tau^+$ is anti-linear.

\begin{thm}\label{thm:QFfermionic}
Let $(M,S,P)$ be an object in $\GlobHypDef$ and $\tau$ a state on the corresponding fermionic algebra $\Aferm(M,S,P)$.
Then
\begin{enumerate}[(i)]
\item\label{ferm:PPsiO}
$P\Phi_\tau=P\Phi_\tau^+=0$ holds in the distributional sense;
\item\label{ferm:CARPsi}
the quantum fields satisfy the canonical anti-commutation relations, i.e., 
\be 
\{\Phi_\tau(f),\Phi_\tau(g)\} &=& \{\Phi_\tau^+(f),\Phi_\tau^+(g)\}=0,\\
\{\Phi_\tau(f),\Phi_\tau^+(g)\} &=&i\Big(\int_M \<Gf,g\>\dV\Big) \cdot \id_{\H_\tau}
\ee
for all $f,g\in\DD(M,S)$.
\end{enumerate}
\end{thm}

\begin{proof}
Since $GP=0$ on $\DD(M,S)$, we have $P\Phi_\tau(f) = \Phi_\tau(Pf) = -\pi_\tau(\a(GPf)^*) = 0$ and similarly for $\Phi_\tau^+$.
This proves assertion \eqref{ferm:PPsiO}.

Using Definition~\ref{def-CAR}~\eqref{def-CAR:2} we compute
\begin{align*}
 \{\Phi_\tau(f),\Phi_\tau(g)\}
&=\{\pi_\tau(\a(Gf)^*),\pi_\tau(\a(Gg)^*)\}\\
&=\pi_\tau(\{\a(Gf)^*,\a(Gg)^*\})\\
&=\pi_\tau(\{\a(Gg),\a(Gf)\}^*)\\
&=0.
\end{align*}
Similarly one sees $\{\Phi_\tau^+(f),\Phi_\tau^+(g)\}=0$.
Definition~\ref{def-CAR}~\eqref{def-CAR:3} also yields
\begin{align*}
 \{\Phi_\tau(f),\Phi_\tau^+(g)\}
&=-\pi_\tau(\{\a(Gf)^*,\a(Gg)\})
=-(Gf,Gg)\cdot\id_{\H_\tau}.
\end{align*}
To prove assertion~\eqref{ferm:CARPsi} we have to verify
\begin{equation}
(Gf,Gg) = -i\int_{M}\<Gf,g\>\dV
\label{eq:M}
\end{equation}
Let $\Sigma\subset M$ be a smooth spacelike Cauchy hypersurface.
Since $\supp(G_+g)$ is past-compact, we can find a Cauchy hypersurface $\Sigma'\subset M$ in the past of $\Sigma$ which does not intersect $\supp(G_+g)\subset J^M_+(\supp(g))$.
Denote the region between $\Sigma$ and $\Sigma'$ by $\Omega'$.
The Green's formula (\ref{eq:Green1storder}) yields
\begin{align*}
(Gf,G_+g)
&=
\int_\Sigma \< i\sigma_P(\mathfrak{n}^\flat)\cdot Gf,G_+g\>\ds \nonumber\\
&=
\int_{\Sigma'} \< i\sigma_P(\mathfrak{n}^\flat)\cdot Gf,G_+g\>\ds 
+ i\int_{\Omega'}(\<PGf,G_+g\>-\<Gf,PG_+g\>)\dV\nonumber\\
&=
-i\int_{\Omega'}\<Gf,g\>\dV
\end{align*}
because $PG_+g=g$ and $PGf=0$.
Since $\Sigma'$ can be chosen arbitrarily to the past, this shows
\begin{equation}
(Gf,G_+g) = -i\int_{J_-(\Sigma)}\<Gf,g\>\dV.
\label{eq:M-}
\end{equation}
A similar computation yields
\begin{equation}
(Gf,G_-g) = i\int_{J_+(\Sigma)}\<Gf,g\>\dV.
\label{eq:M+}
\end{equation}
Subtracting \eqref{eq:M+} from \eqref{eq:M-} yields \eqref{eq:M} and concludes the proof of assertion~\eqref{ferm:CARPsi}.
\end{proof}

\begin{rem}
Similarly to the bosonic case, we find
$$
\{\Phi_\tau(f),\Phi_\tau^+(g)\}=0
$$
if the supports of $f$ and $g$ are causally disjoint.
\end{rem}

\begin{rem}
Using the anti-commutation relations in Theorem~\ref{thm:QFfermionic}~\eqref{ferm:CARPsi}, the computation of $n$-point functions can be reduced to those of the form
$$
\tau_{n,n'}(f_1,\ldots,f_n,g_1,\ldots,g_{n'}) = \<\Omega_\tau,\Phi_\tau(f_1)\cdots\Phi_\tau(f_n)\Phi_\tau^+(g_1)\cdots\Phi_\tau^+(g_{n'})\Omega_\tau\>_\tau .
$$
As in the bosonic case, the $n$-point functions satisfy the field equation in the distributional sense in each argument and extend to distributions on $M\times\cdots\times M$.
\end{rem}

If one uses the self-dual fermionic algebra $\Afermsd(M,S,P)$ instead of $\Aferm(M,S,P)$, then one gets the quantum field
$$
\Psi_\tau(f) := \pi_\tau(\b(Gf))
$$
where $\b$ is as in Definition~\ref{def-sdCAR}.
Then the analogue to Theorem~\ref{thm:QFfermionic} is

\begin{thm}
Let $(M,S,P)$ be an object in $\GHSD$ and $\tau$ a state on the corresponding self-dual fermionic algebra $\Afermsd(M,S,P)$.
Then
\begin{enumerate}[(i)]
\item
$P\Psi_\tau=0$ holds in the distributional sense;
\item
the quantum field takes values in self-adjoint operators, $\Psi_\tau(f) = \Psi_\tau(f)^*$ for all $f\in\DD(M,S)$;
\item
the quantum fields satisfy the canonical anti-commutation relations, i.e., 
$$
\{\Psi_\tau(f),\Psi_\tau(g)\} 
=
\int_M \<Gf,g\>\dV \cdot \id_{\H_\tau}
$$
for all $f,g\in\DD(M,S)$.
\end{enumerate}
\end{thm}

\begin{rem}
It is interesting to compare the concept of locally covariant quantum field theories as proposed in \cite{BFV} to the axiomatic approach to quantum field theory on Minkowski space based on the G\r{a}rding-Wightman axioms as exposed in \cite[Sec.~IX.8]{RS2}.
Property~1 (relativistic invariance of states) and Property~6 (Poincar\'e invariance of the field) in \cite{RS2} are replaced by functoriality (covariance). 
Property~4 (invariant domain for fields) and Property~5 (regularity of the field) have been encoded in strong regularity of the state used to define the quantum field in the bosonic case and are automatic in the fermionic case.
Property~7 (local commutativity or microscopic causality) is contained in Theorems~\ref{thm:QFbosonic} and \ref{thm:QFfermionic}.
Property~3 (existence and uniqueness of the vacuum) has no analogue and is replaced by the \emph{choice} of a state.
Property~8 (cyclicity of the vacuum) is then automatic by the general properties of the GNS-construction.

There remains one axiom, Property~2 (spectral condition), which we have not discussed at all.
It gets replaced by the Hadamard condition on the state chosen.
It was observed by Radzikowski \cite{Rad} that earlier formulations of this condition are equivalent to a condition on the wave front set of the $2$-point function.
Much work has been put into constructing and investigating Hadamard states for various examples of fields, see e.g.\ \cite{DMP,DPP,FewVerch,HollWald,SV1,SV2,S,Wald} and the references therein.
\end{rem}

\appendix
\section{Algebras of canonical (anti-) commutation relations}

We collect the necessary algebraic facts about CAR and CCR-algebras.

\subsection{CAR algebras}\label{s:appendixCAR}
The symbol ``CAR'' stands for ``canonical anti-commutation relations''.
These algebras are related to pre-Hilbert spaces.
We always assume the Hermitian inner product $(\cdot\,,\cdot)$ to be linear in the first argument and anti-linear in the second.

\begin{definition}\label{def-CAR}
A \emph{$\mathrm{CAR}$-representation} of a complex pre-Hilbert space $(V,(\cdot\,,\cdot))$ is a pair $(\a,A)$, where $A$ is a unital \Cstar-algebra and $\a:V\to A$ is an anti-linear map satisfying:
\begin{enumerate}[(i)]
\item\label{def-CAR:1} $A=C^*(\a(V))$,
\item\label{def-CAR:2} $\{\a(v_1),\a(v_2)\}=0$ and
\item\label{def-CAR:3} $\{\a(v_1)^*,\a(v_2)\}=(v_1,v_2)\cdot 1$,
\end{enumerate}
for all $v_1,v_2\in V$.
\end{definition}

We want to discuss CAR-representations in terms of \Cstar-Clifford algebras, whose definition we recall.
Given a complex pre-Hilbert vector space $(V,(\cdot\,,\cdot))$, we denote by $V_\C:=V\otimes_\R\C$ the complexification of $V$ considered as a real vector space and by $q_\C$ the complex-bilinear extension of $\Re(\cdot\,,\cdot)$ to $V_\C$.
Let $\mathrm{Cl}_{\rm alg}(V_\C,q_\C)$ be the algebraic Clifford algebra of  $(V_\C,q_\C)$.
It is an associative complex algebra with unit and contains $V_\C$ as a vector subspace.
Its multiplication is called Clifford multiplication and denoted by ``$\,\cdot\,$''. 
It satisfies the Clifford relations
\begin{equation} 
v\cdot w+w\cdot v=-2q_\C(v,w)1
\label{eq:CliffRel}
\end{equation}
for all $v,w\in V_\C$.
Define the $*$-operator on $\mathrm{Cl}_{\rm alg}(V_\C,q_\C)$ to be the unique anti-multiplicative and anti-linear extension of the anti-linear map $V_\C\rightarrow V_\C$, $v_1+iv_2\mapsto -(\overline{v_1+iv_2})=-(v_1-iv_2)$ for all $v_1,v_2\in V$.
In other words,
\[
*(\sum_{i_1<\ldots<i_k}\alpha_{i_1,\ldots,i_k}z_{i_1}\cdot\ldots\cdot z_{i_k})
=
(-1)^k\sum_{i_1<\ldots<i_k}\overline{\alpha_{i_1,\ldots,i_k}}\cdot\overline{z_{i_k}}\cdot\ldots\cdot \overline{z_{i_1}}
\]
for all $k\in\mathbb{N}$ and $z_{i_1},\ldots,z_{i_k}\in V_\C$.
Let $\|\cdot\|_\infty$ be defined by 
\[ \|a\|_\infty:=\sup_{\pi\in\mathrm{Rep}(V)}(\|\pi(a)\|)\]
for every $a\in\mathrm{Cl}_{\rm alg}(V_\C,q_\C)$, where $\mathrm{Rep}(V)$ denotes the set of all (isomorphism classes of) $*$-homomorphisms from $\mathrm{Cl}_{\rm alg}(V_\C,q_\C)$ to \Cstar-algebras.
Then $\|\cdot\|_\infty$ can be shown to be a well-defined \Cstar-norm on $\mathrm{Cl}_{\rm alg}(V_\C,q_\C)$, see e.g.\ \cite[Sec.~1.2]{PR}.

\begin{definition}\label{def-C*Cliffalg}
The \Cstar-Clifford algebra of a pre-Hilbert space $(V,(\cdot\,,\cdot))$ is the \Cstar-com\-ple\-tion of $\mathrm{Cl}_{\rm alg}(V_\C,q_\C)$ with respect to the \Cstar-norm $\|\cdot\|_\infty$ and the star operator defined above.
\end{definition}

\begin{thm}\label{tCAR}
For every complex pre-Hilbert space $(V,(\cdot\,,\cdot))$, the \Cstar-Clifford algebra $\mathrm{Cl}(V_\C,q_\C)$ provides a {\rm CAR}-representation of $(V,(\cdot\,,\cdot))$ via $\a(v)=\frac{1}{2}(v+iJv)$, where $J$ is the complex structure of $V$.

Moreover, {\rm CAR}-representations have the following universal property:
Let $\wih{A}$ be any unital \Cstar-algebra and $\wih{\a}:V\to\wih{A}$ be any anti-linear map satisfying Axioms~\eqref{def-CAR:2} and \eqref{def-CAR:3} of Definition~\ref{def-CAR}.
Then there exists a unique \Cstar-morphism $\wit{\alpha}: \mathrm{Cl}(V_\C,q_\C)\to\wih{A}$ such that 
$$
\xymatrix{
V \ar[r]^{\wih{\a}} \ar[d]_\a & \wih{A}\\
\mathrm{Cl}(V_\C,q_\C) \ar@{.>}[ru]^{\wit{\alpha}} &}
$$
commutes.
Furthermore, $\wit{\alpha}$ is injective.
\end{thm}

\begin{proof}
Define $p_\mp : V\to\mathrm{Cl}(V_\C,q_\C)$ by $p_-(v):=\frac{1}{2}(v+iJv)$ and $p_+(v):=\frac{1}{2}(v-iJv)$.
Since $p_-(Jv)=-ip_-(v)$, the map $\a=p_-$ is anti-linear.
Because of $\a(v)-\a(v)^*=p_-(v)+p_+(v)=v$, the \Cstar-subalgebra of $\mathrm{Cl}(V_\C,q_\C)$ generated by the image of $\a$ contains $V$.
Hence $\a(V)$ generates $\mathrm{Cl}(V_\C,q_\C)$ as a \Cstar-algebra.
Axiom~\eqref{def-CAR:1} in Definition~\ref{def-CAR} is proved.

Let $v_1,v_2\in V$, then
\be 
\{\a(v_1),\a(v_2)\}
&=&
p_-(v_1)\cdot p_-(v_2)+p_-(v_2)\cdot p_-(v_1)\\
&=&
-2q_\C(p_-(v_1), p_-(v_2))\cdot 1\\
&=&0,
\ee
which is Axiom~\eqref{def-CAR:2} in Definition~\ref{def-CAR}.
Furthermore, 
\be 
\{\a(v_1)^*,\a(v_2)\}&=&-p_+(v_1)\cdot p_-(v_2)-p_-(v_2)\cdot p_+(v_1)\\
&=&2q_\C(p_+(v_1),p_-(v_2))\cdot 1\\
&=&\Re (v_1,v_2)\cdot 1+i\Re (v_1,Jv_2)\cdot 1\\
&=&(v_1,v_2)\cdot 1,
\ee 
which shows Axiom~\eqref{def-CAR:3} in Definition~\ref{def-CAR}.
Therefore $(\a,\mathrm{Cl}(V_\C,q_\C))$ is a CAR-representation of $(V,(\cdot\,,\cdot))$.

The second part of the theorem follows from $\mathrm{Cl}(V_\C,q_\C)$ being simple, i.e., from the non-existence of non-trivial closed two-sided $*$-invariant ideals,
see \cite[Thm.~1.2.2]{PR}.
%
Let $\wih{\a}:V\rightarrow \wih{A}$ be any other anti-linear map satisfying \eqref{def-CAR:2} and \eqref{def-CAR:3} in Definition~\ref{def-CAR}.
Since $\a$ and $\wih{\a}$ are injective (which is clear by Axiom~\eqref{def-CAR:3}) one may set $\alpha(\a(v)):=\wih{\a}(v)$ for all $v\in V$.
Axioms~\eqref{def-CAR:2} and \eqref{def-CAR:3} allow us to extend $\alpha$ to a \Cstar-morphism $\wit{\alpha}:C^*(\a(V))=\mathrm{Cl}(V_\C,q_\C)\to\wih{A}$.
The injectivity of $\wih{\a}$ implies the non-triviality of $\wit{\alpha}$ which, together with the simplicity of $\mathrm{Cl}(V_\C,q_\C)$, provides the injectivity of $\wit{\alpha}$.
Therefore we found an injective \Cstar-morphism $\wit{\alpha}:\mathrm{Cl}(V_\C,q_\C)\to\wih{A}$ with $\wit{\alpha}\circ \a=\wih{\a}$.
It is unique since it is determined by $\a$ and $\wih{\a}$ on a subset of generators.
This concludes the proof of Theorem~\ref{tCAR}.
\end{proof}

For an alternative description of the CAR-representation in terms of creation and annihilation operators on the fermionic Fock space we refer to \cite[Prop.~5.2.2]{BR}.

\begin{cor}\label{cCAR}
For every complex pre-Hilbert space $(V,(\cdot\,,\cdot))$ there exists a ${\rm CAR}$-representation of $(V,(\cdot\,,\cdot))$, unique up to \Cstar-isomorphism.
\end{cor}

\begin{proof}
The existence has already been proved in Theorem~\ref{tCAR}.
Let $(\wih{\a},\wih{A})$ be any CAR-representation of $(V,(\cdot\,,\cdot))$.
Theorem~\ref{tCAR} states the existence of a unique injective \Cstar-morphism $\wit{\alpha}: \mathrm{Cl}(V_\C,q_\C)\to\wih{A}$ such that $\wit{\alpha}\circ \a=\wih{\a}$.
Now $\wit{\alpha}$ has to be surjective since Axiom~\eqref{def-CAR:1} holds for $(\wih{\a},\wih{A})$.
\end{proof}

From now on, given a complex pre-Hilbert space $(V,(\cdot\,,\cdot))$, we denote  the \Cstar-algebra $\mathrm{Cl}(V_\C,q_\C)$ associated with the $\mathrm{CAR}$-representation $(\a,\mathrm{Cl}(V_\C,q_\C))$ of $(V,(\cdot\,,\cdot))$ by $\mathrm{CAR}(V,(\cdot\,,\cdot))$.
We list the properties of CAR-representations which are relevant for quantization, see also \cite[Vol.~II, Thm.~5.2.5, p.~15]{BR}.

\begin{prop}\label{pCAR}
Let $(V,(\cdot\,,\cdot))$ be a complex pre-Hilbert space and  $(\a,\mathrm{CAR}(V,(\cdot\,,\cdot)))$ its {\rm CAR}-representation.
\begin{enumerate}[(i)]
\item\label{pCAR:1}
For every $v\in V$ one has $\|\a(v)\|=|v|=(v,v)^{\frac{1}{2}}$, where $\|\cdot\|$ denotes the \Cstar-norm on $\mathrm{CAR}(V,(\cdot\,,\cdot))$.
\item \label{pCAR:2}
The \Cstar-algebra $\mathrm{CAR}(V,(\cdot\,,\cdot))$ is simple, i.e., it has no closed two-sided $*$-ideals other than $\{0\}$ and the algebra itself.
\item \label{pCAR:4}
The algebra $\mathrm{CAR}(V,(\cdot\,,\cdot))$ is $\Z_2$-graded, 
$$
\mathrm{CAR}(V,(\cdot\,,\cdot))=\mathrm{CAR}^\mathrm{even}(V,(\cdot\,,\cdot))\oplus\mathrm{CAR}^\mathrm{odd}(V,(\cdot\,,\cdot)),
$$ 
and $\a(V) \subset \mathrm{CAR}^\mathrm{odd}(V,(\cdot\,,\cdot))$.
\item \label{pCAR:3}
Let $f:V\to V'$ be an isometric linear embedding, where $(V',(\cdot\,,\cdot)')$ is another complex pre-Hilbert space.
Then there exists a unique injective \Cstar-morphism $\mathrm{CAR}(f):\mathrm{CAR}(V,(\cdot\,,\cdot))\to\mathrm{CAR}(V',(\cdot\,,\cdot)')$ such that 
$$
\xymatrix{
V\ar[rr]^{f}\ar[d]^{\a}
&&V' \ar[d]^{\a'}\\
\mathrm{CAR}(V,(\cdot\,,\cdot))\ar[rr]^{\mathrm{CAR}(f)}
&&\mathrm{CAR}(V',(\cdot\,,\cdot)')
}
$$
commutes.
\end{enumerate}
\end{prop}

\begin{proof}
We show assertion~\eqref{pCAR:1} .
On the one hand, the \Cstar-property of the norm $\|\cdot\|$ implies
\be  
\|\a(v)\|^4
&=&\|\a(v)\a(v)^*\|^2\\
&=&\|(\a(v)\a(v)^*)^2\|.
\ee
On the other hand, 
\be  
(\a(v)\a(v)^*)^2&=&\a(v)\{\a(v)^*,\a(v)\}\a(v)^*\\
&=&|v|^2\a(v)\a(v)^*,
\ee
where we used $\a(v)^2=0$ which follows from the second axiom.
We deduce that
\be 
\|\a(v)\|^4&=&|v|^2\cdot\|\a(v)\a(v)^*\|\\
&=&|v|^2\cdot\|\a(v)\|^2.
\ee
Since $\a$ is injective, we obtain the result.

Assertion~\eqref{pCAR:2} follows from $\mathrm{Cl}(V_\C,q_\C)$ being simple, see \cite[Thm.~1.2.2]{PR}.
Alternatively, it can be deduced from the universal property formulated in Theorem~\ref{tCAR}.

To see \eqref{pCAR:4} we recall that the Clifford algebra $\mathrm{Cl}(V_\C,q_\C)$ has a $\Z_2$-grading where the even part is generated by products of an even number of vectors in $V_\C$ and, similarly, the odd part is the vector space span of products of an odd number of vectors in $V_\C$, see \cite[p.~27]{PR}.
This is compatible with the Clifford relations \eqref{eq:CliffRel}.
Clearly, $\a(V) \subset \mathrm{CAR}^\mathrm{odd}(V,(\cdot\,,\cdot))$.

It remains to show \eqref{pCAR:3}.
It is straightforward to check that $\a'\circ f$ satisfies Axioms~\eqref{def-CAR:2} and \eqref{def-CAR:3} in Definition~\ref{def-CAR}.
The result follows from Theorem~\ref{tCAR}.
\end{proof}

One easily sees that $\mathrm{CAR}(\id)=\id$ and that $\mathrm{CAR}(f'\circ f)=\mathrm{CAR}(f')\circ\mathrm{CAR}(f)$ for all isometric linear embeddings $V\xrightarrow{f}V'\xrightarrow{f'}V''$.
Therefore we have constructed a covariant functor
\[\mathrm{CAR}:\hilb\longrightarrow\CAlg,\]
where $\hilb$ denotes the category whose objects are the complex pre-Hilbert spaces and whose morphisms are the isometric linear embeddings.

For {\em real} pre-Hilbert spaces there is the concept of {\em self-dual} CAR-representations. 

\begin{definition}\label{def-sdCAR}
A \emph{self-dual $\mathrm{CAR}$-representation} of a real pre-Hilbert space $(V,(\cdot\,,\cdot))$ is a pair $(\b,A)$, where $A$ is a unital \Cstar-algebra and $\b:V\to A$ is an $\R$-linear map satisfying:
\begin{enumerate}[(i)]
\item\label{def-sdCAR:1} $A=C^*(\b(V))$,
\item\label{def-sdCAR:2} $\b(v)=\b(v)^*$ and
\item\label{def-sdCAR:3} $\{\b(v_1),\b(v_2)\}=(v_1,v_2)\cdot 1$,
\end{enumerate}
for all $v,v_1,v_2\in V$.
\end{definition}
Given a self-dual CAR-representation, one can extend $\b$ to a $\C$-linear map from the complexification $V_\C$ to $A$.
This extension $\b:V_\C\to A$ then satisfies $\b(\bar v)=\b(v)^*$ and $\{\b(v_1),\b(v_2)\}=(v_1,\bar v_2)\cdot 1$ for all $v,v_1,v_2\in V_\C$.
These are the axioms of a self-dual CAR-representation as in \cite[p.~386]{Ar}.

\begin{thm}\label{tsdCAR}
For every real pre-Hilbert space $(V,(\cdot\,,\cdot))$, the \Cstar-Clifford algebra $\mathrm{Cl}(V_\C,q_\C)$ provides a self-dual {\rm CAR}-representation of $(V,(\cdot\,,\cdot))$ via $\b(v)=\frac{i}{\sqrt 2}v$.

Moreover, self-dual {\rm CAR}-representations have the following universal property:
Let $\wih{A}$ be any unital \Cstar-algebra and $\wih{\b}:V\to\wih{A}$ be any $\R$-linear map satisfying Axioms~\eqref{def-sdCAR:2} and \eqref{def-sdCAR:3} of Definition~\ref{def-sdCAR}.
Then there exists a unique \Cstar-morphism $\wit{\beta}: \mathrm{Cl}(V_\C,q_\C)\to\wih{A}$ such that 
$$
\xymatrix{
V \ar[r]^{\wih{\b}} \ar[d]_\b & \wih{A}\\
\mathrm{Cl}(V_\C,q_\C) \ar@{.>}[ru]^{\wit{\beta}} &}
$$
commutes.
Furthermore, $\wit{\beta}$ is injective.
\end{thm}

\begin{cor}\label{csdCAR}
For every real pre-Hilbert space $(V,(\cdot\,,\cdot))$ there exists a ${\rm CAR}$-representation of $(V,(\cdot\,,\cdot))$, unique up to \Cstar-isomorphism.
\end{cor}

From now on, given a real pre-Hilbert space $(V,(\cdot\,,\cdot))$, we denote  the \Cstar-algebra $\mathrm{Cl}(V_\C,q_\C)$ associated with the self-dual $\mathrm{CAR}$-representation $(\b,\mathrm{Cl}(V_\C,q_\C))$ of $(V,(\cdot\,,\cdot))$ by $\mathrm{CAR}_\mathrm{sd}(V,(\cdot\,,\cdot))$.

\begin{prop}\label{psdCAR}
Let $(V,(\cdot\,,\cdot))$ be a real pre-Hilbert space and  $(\b,\mathrm{CAR}_\mathrm{sd}(V,(\cdot\,,\cdot)))$ its self-dual {\rm CAR}-representation.
\begin{enumerate}[(i)]
\item\label{psdCAR:1}
For every $v\in V$ one has $\|\b(v)\|=\frac{1}{\sqrt 2}|v|$, where $\|\cdot\|$ denotes the \Cstar-norm on $\mathrm{CAR}_\mathrm{sd}(V,(\cdot\,,\cdot))$.
\item \label{psdCAR:2}
The \Cstar-algebra $\mathrm{CAR}_\mathrm{sd}(V,(\cdot\,,\cdot))$ is simple.
\item \label{psdCAR:4}
The algebra $\mathrm{CAR}_\mathrm{sd}(V,(\cdot\,,\cdot))$ is $\Z_2$-graded, 
$$
\mathrm{CAR}_\mathrm{sd}(V,(\cdot\,,\cdot))=\mathrm{CAR}^\mathrm{even}_\mathrm{sd}(V,(\cdot\,,\cdot))\oplus\mathrm{CAR}^\mathrm{odd}_\mathrm{sd}(V,(\cdot\,,\cdot)),
$$ 
and $\b(V) \subset \mathrm{CAR}^\mathrm{odd}_\mathrm{sd}(V,(\cdot\,,\cdot))$.
\item \label{psdCAR:3}
Let $f:V\to V'$ be an isometric linear embedding, where $(V',(\cdot\,,\cdot)')$ is another real pre-Hilbert space.
Then there exists a unique injective \Cstar-morphism $\mathrm{CAR}_\mathrm{sd}(f):\mathrm{CAR}_\mathrm{sd}(V,(\cdot\,,\cdot))\to\mathrm{CAR}_\mathrm{sd}(V',(\cdot\,,\cdot)')$ such that 
$$
\xymatrix{
V\ar[rr]^{f}\ar[d]^{\b}
&&V' \ar[d]^{\b'}\\
\mathrm{CAR}_\mathrm{sd}(V,(\cdot\,,\cdot))\ar[rr]^{\mathrm{CAR}_\mathrm{sd}(f)}
&&\mathrm{CAR}_\mathrm{sd}(V',(\cdot\,,\cdot)')
}
$$
commutes.
\end{enumerate}
\end{prop}

The proofs are similar to the ones for CAR-representations of complex pre-Hilbert spaces.
We have constructed a functor
\[\mathrm{CAR}_\mathrm{sd}:\rhilb\longrightarrow\CAlg,\]
where $\rhilb$ denotes the category whose objects are the real pre-Hilbert spaces and whose morphisms are the isometric linear embeddings.

\begin{rem}\label{rem:RversusC}
Let $(V,(\cdot\,,\cdot))$ be a complex pre-Hilbert space.
If we consider $V$ as a real vector space, then we have the real pre-Hilbert space $(V,\Re(\cdot\,,\cdot))$.
For the corresponding CAR-representations we have
$$
\mathrm{CAR}(V,(\cdot\,,\cdot)) 
= \mathrm{CAR}_\mathrm{sd}(V,\Re(\cdot\,,\cdot))
= \mathrm{Cl}(V_\C,q_\C)
$$
and
$$
\b(v) = \frac{i}{\sqrt 2}(\a(v)-\a(v)^*) .
$$
\end{rem}

\subsection{CCR algebras}\label{s:appendixCCR}

In this section, we recall the construction of the representation of any (real) symplectic vector space by the so-called canonical commutation relations (CCR). 
Proofs can be found in \cite[Sec.~4.2]{BGP}.

\begin{definition}\label{d:CCR}
A \emph{{\rm CCR}-representation} of a symplectic vector space $(V,\omega)$ is a pair $(w,A)$, where $A$ is a unital \Cstar-algebra and $w$ is a map $V\to A$ satisfying:
\begin{enumerate}[(i)]
\item\label{CCR1} 
$A=C^*(w(V))$,
\item\label{CCR2} 
$w(0)=1$,
\item\label{CCR3}  
$w(-\phi)=w(\phi)^*$,
\item\label{CCR4}  
$w(\phi+\psi)=e^{i\omega(\phi,\psi)/2}w(\phi)\cdot w(\psi)$,
\end{enumerate}
for all $\phi,\psi\in V$.
\end{definition}

The map $w$ is in general neither linear, nor any kind of group homomorphism, nor continuous \cite[Prop. 4.2.3]{BGP}.

\begin{ex}\label{ex:CCR}
Given any symplectic vector space $(V,\omega)$, consider the Hilbert space $H:=L^2(V,\mathbb{C})$, where $V$ is endowed with the counting measure.
Define the map $w$ from $V$ into the space $\mathcal{L}(H)$ of bounded endomorphisms of $H$ by 
\[(w(\phi)F)(\psi):=e^{i\omega(\phi,\psi)/2}F(\phi+\psi),\]
for all $\phi,\psi\in V$ and $F\in H$. 
It is well-known that $\mathcal{L}(H)$ is a \Cstar-algebra with the operator norm as \Cstar-norm, and that the map $w$ satisfies the Axioms~\eqref{CCR2}-\eqref{CCR4} from Definition~\ref{d:CCR}, see e.g.\ \cite[Ex.~4.2.2]{BGP}. 
Hence setting $A:=C^*(w(V))$, the pair $(w,A)$ provides a {\rm CCR}-representation of $(V,\omega)$.
\end{ex}

This is essentially the only example of {\rm CCR}-representation:

\begin{thm}\label{t:existuniqueCCR}
Let $(V,\omega)$ be a symplectic vector space and $(\hat{w},\wih{A})$ be a pair satisfying the Axioms~\eqref{CCR2}-\eqref{CCR4} of Definition~\ref{d:CCR}.
Then there exists a unique \Cstar-morphism $\Phi:A\rightarrow\wih{A}$ such that $\Phi\circ w=\hat{w}$, where $(w,A)$ is the {\rm CCR}-representation from Example~\ref{ex:CCR}.
Moreover, $\Phi$ is injective.

In particular, $(V,\omega)$ has a {\rm CCR}-representation, unique up to \Cstar-isomorphism.
\end{thm}

We denote the \Cstar-algebra associated to the {\rm CCR}-representation of $(V,\omega)$ from Example~\ref{ex:CCR} by $\CCR(V,\omega)$. 
As a consequence of Theorem~\ref{t:existuniqueCCR}, we obtain the following important corollary.

\begin{cor}\label{c:CCR}
Let $(V,\omega)$ be a symplectic vector space and $(w,\CCR(V,\omega))$ its $\CCR$-representation.
\begin{enumerate}[(i)]
\item
The \Cstar-algebra $\CCR(V,\omega)$ is simple, i.e., it has no closed two-sided $*$-ideals other than $\{0\}$ and the algebra itself.
\item 
Let $(V',\omega')$ be another symplectic vector space and $f:V\to V'$ a symplectic linear map.
Then there exists a unique injective \Cstar-morphism $\CCR(f):\CCR(V,\omega)\rightarrow\CCR(V',\omega')$ such that 
$$
\xymatrix{
V\ar[rr]^{f}\ar[d]^{w}
&&V' \ar[d]^{w'}\\
\CCR(V,\omega)\ar[rr]^{\CCR(f)}
&&\CCR(V',\omega')
}
$$
commutes.
\end{enumerate}
\end{cor}

Obviously $\CCR(\id)=\id$ and $\CCR(f'\circ f)=\CCR(f')\circ\CCR(f)$ for all symplectic linear maps $V\bui{\rightarrow}{f} V'\bui{\rightarrow}{f'} V''$, so that we have constructed a covariant functor
\[\CCR:\Sympl\longrightarrow\CAlg.\]

\end{document}